\documentclass[sigconf,authorversion]{acmart}

\usepackage{stmaryrd} %
\usepackage[utf8]{inputenc}

\usepackage{xspace}
\usepackage{subcaption} %
\usepackage{tabularx} %
\usepackage{booktabs} %

\def\polylog{\operatorname{polylog}} %

\newcommand{\countlessthan}[2]{\downarrow_{#2}\!\!(#1)} 
\newcommand{\APPROXCOUNT}{\textsc{ApxCnt}\xspace}
\newcommand{\JQs}{JQs\xspace}
\newcommand{\JQ}{JQ\xspace}
\newcommand{\quant}{\phi}
\newcommand{\qanswer}{q}
\newcommand{\Wvars}{{U_\w}}
\newcommand{\wval}{\lambda}
\newcommand{\Waggr}{\mathtt{agg}_\w}
\newcommand{\Wdom}{\dom_\w}

\newcommand{\sumset}{\sigma}
\newcommand{\w}{w}
\newcommand{\rankf}{subset-monotone\xspace}

\newcommand{\TRIM}{\textsc{trim}\xspace}
\newcommand{\PIVOT}{\textsc{pivot}\xspace}

\newcommand{\hyperclique}{\textsc{Hyperclique}}
\newcommand{\ThreeSUM}{\textsc{3sum}}

\newcommand{\calH}{{\mathcal H}}
\newcommand{\calS}{{\mathcal S}}
\newcommand{\var}{\texttt{var}}

\newcommand{\bigO}{{\mathcal{O}}} %
\newcommand{\mh}{\textrm{mh}}

\newcommand{\dom}{\texttt{dom}}
\newcommand{\ar}{\texttt{ar}}

\newcommand{\wm}{\textsc{wmed}}
\newcommand{\pivot}{\textrm{pivot}}

\def \cnt{\texttt{cnt}}
\def \mval{\texttt{val}}

\newcommand{\N}{\mathbb{N}} %
\newcommand{\R}{{\mathbb{R}}} %

\def\e#1{\emph{#1}}

\newcommand{\introparagraph}[1]{\textbf{#1.}} %
\newcommand{\datarule}{{\,:\!\!-\,}} %
\newcommand{\nikos}[1]{{{\color{blue}{[\textbf{Nikos}: #1]}}}} %
\newcommand{\nikosdone}[1]{} %
\newcommand{\nofar}[1]{{{\color{violet}{[\textbf{Nofar}: #1]}}}} %
\newcommand{\nofardoe}[1]{} %
\newcommand{\benny}[1]{{{\color{olive}{[\textbf{Benny}: #1]}}}} %
\newcommand{\wolf}[1]{{{\color{magenta}{[\textbf{WG}: #1]}}}} %
\newcommand{\wolfdone}[1]{} %
\newcommand{\mirek}[1]{{{\color{cyan}{[\textbf{MR}: #1]}}}} %
\newcommand\reviewer[2]{{\color{red}[\textbf{Reviewer #1:} #2]}}
\newcommand{\nikosrem}[1]{} 
\newcommand{\nofarrem}[1]{}
\newcommand{\wolfrem}[1]{} %
\newcommand{\mirekrem}[1]{}
\newcommand{\todo}[1]{{{\color{red}{[TODO: #1]}}}} %

\newcommand{\ourversion}[1]{#1} 
\newcommand{\hide}[1]{} 
\newcommand{\hidetwo}[2]{}

\newcommand{\mkclean}{

    \renewcommand{\ourversion}{\hide}

    \renewcommand{\mirek}{\hide}
    \renewcommand{\wolf}{\hide}
    \renewcommand{\nikos}{\hide}
    \renewcommand{\nofar}{\hide}
    \renewcommand{\benny}{\hide}
    \renewcommand{\reviewer}{\hidetwo}
    \renewcommand{\todo}{\hide}
}

\mkclean %

\def\rel#1{\mathsf{#1}}
\def\att#1{\mathrm{#1}}

\newcommand{\eat}[1]{}

\usepackage[ruled,noend,linesnumbered]{algorithm2e} %

\usepackage{setspace} %

\DontPrintSemicolon     %
\SetNlSty{}{}{}                %
\SetAlgoInsideSkip{smallskip}   %

\SetAlFnt{\small}			%
\SetAlCapFnt{\small}		%
\SetAlCapNameFnt{\small}

\newcommand{\algocomment}[1]{\textcolor{teal}{{//#1}}}

\usepackage[capitalise,nameinlink]{cleveref} %
\crefformat{footnote}{#2\footnotemark[#1]#3}

\AtEndPreamble{%
    \hypersetup{colorlinks,
      linkcolor=purple,
      citecolor=blue,
      urlcolor=ACMDarkBlue,
      filecolor=ACMDarkBlue}}

\def\ranklex{\ensuremath{\mathsf{LEX}}\xspace}
\def\rankmin{\ensuremath{\mathsf{MIN}}\xspace}
\def\rankmax{\ensuremath{\mathsf{MAX}}\xspace}
\def\ranksum{\ensuremath{\mathsf{SUM}}\xspace}

\AtBeginDocument{%
  }

\copyrightyear{2023} 
\acmYear{2023} 
\setcopyright{acmlicensed}\acmConference[PODS '23]{Proceedings of the 42nd ACM SIGMOD-SIGACT-SIGAI Symposium on Principles of Database Systems}{June 18--23, 2023}{Seattle, WA, USA}
\acmBooktitle{Proceedings of the 42nd ACM SIGMOD-SIGACT-SIGAI Symposium on Principles of Database Systems (PODS '23), June 18--23, 2023, Seattle, WA, USA}
\acmPrice{15.00}
\acmDOI{10.1145/3584372.3588670}
\acmISBN{979-8-4007-0127-6/23/06}

\begin{document}

\title{Efficient Computation of Quantiles over Joins}

\author{Nikolaos	Tziavelis}
\orcid{0000-0001-8342-2177}
\affiliation{%
  \institution{Northeastern University}
  \city{Boston}
  \state{MA}
  \country{United States}
}
\email{tziavelis.n@northeastern.edu}

\author{Nofar	Carmeli}
\orcid{0000-0003-0673-5510}
\affiliation{%
  \institution{Inria, LIRMM, Univ Montpellier, CNRS}
  \city{Montpellier}
  \country{France}
}
\email{nofar.carmeli@inria.fr}

\author{Wolfgang	Gatterbauer}
\orcid{0000-0002-9614-0504}
\affiliation{%
  \institution{Northeastern University}
  \city{Boston}
 \state{MA}
 \country{United States}
}
\email{w.gatterbauer@northeastern.edu}

\author{Benny Kimelfeld}
\orcid{0000-0002-7156-1572}
\affiliation{%
  \institution{Technion - Israel Institute of Technology}
  \city{Haifa}
  \country{Israel}
}
\email{bennyk@technion.ac.il}

\author{Mirek	Riedewald}
\orcid{0000-0002-6102-7472}
\affiliation{%
  \institution{Northeastern University}
  \city{Boston}
 \state{MA}
  \country{United States}
}
\email{m.riedewald@northeastern.edu}

\renewcommand{\shortauthors}{Nikolaos Tziavelis, Nofar Carmeli, Wolfgang Gatterbauer, Benny Kimelfeld, \& Mirek Riedewald} %

\begin{abstract}
We present efficient algorithms for Quantile Join Queries,
abbreviated as \%\JQ. A \%\JQ asks for the answer at a specified
relative position (e.g., 50\% for the median)
under some ordering over the answers to a Join Query (\JQ).
Our goal is to avoid materializing the set of all join answers,
and to achieve quasilinear time in the size of the database,
regardless of the total number of answers.
A recent dichotomy result rules out the existence of such an algorithm for
a general family of queries and orders.
Specifically, for acyclic \JQs without self-joins,
the problem becomes intractable 
for ordering by sum whenever
we join more than two relations (and these joins are not trivial intersections).
Moreover, even for basic ranking functions beyond sum, such as min or max 
over different attributes, so far it is not known
whether there is any nontrivial tractable \%\JQ.

In this work, we develop a new approach to solving \%\JQ and show how this
approach allows not just to recover known results,
but also generalize them and resolve open cases.
Our solution uses two subroutines: 
The first one needs to select what we call a ``pivot answer''.
The second subroutine partitions the space of query answers
according to this pivot, and continues searching in one partition that is
represented as new \%\JQ over a new database.
For pivot selection, we develop an algorithm that
works for a large class of ranking functions that are appropriately monotone.
The second subroutine requires a customized construction for the
specific ranking function at hand. 

We show the benefit and generality of our approach by using it to establish several new complexity results. 
First, we prove the tractability of min and max for all acyclic \JQs, thereby resolving the above question. 
Second, we extend the previous \%JQ
dichotomy for sum to all partial sums (over all subsets of the attributes). 
Third, we handle the intractable cases of sum 
by devising a \emph{deterministic} approximation scheme that applies to \emph{every acyclic \JQ}. 

\end{abstract}



\maketitle

\section{Introduction}

Quantile queries ask for the element at a given relative position $\quant\in[0,1]$ in a given list $L$ of items~\cite{rice2007msd}. For example, the lower quartile, median, and upper quartile are the elements at positions $\quant=0.25$, $\quant=0.5$, and $\quant=0.75$, respectively.
We investigate quantile queries where $L$ is the result of a Join Query (\JQ) $Q$ over a database, with a ranking function that determines the order between the answers. 
Importantly, the list $L$ can be much larger than the input database $D$ (specifically, $L$ can be $\Omega(|D|^k)$ for some degree $k$ determined by $Q$), 
and so, $Q$ and $D$ form a compact representation for $L$. 
Our main research question is \e{when we can find the quantile in quasilinear time}.
In other words, the time suffices for reading $D$, but we are generally prevented
from materializing $Q(D)$.

For illustration, consider a social network where users organize events, share event announcements, and declare their plans to attend events. It has the three relations
$\rel{Admin}(\att{user},\att{event})$,
$\rel{Share}(\att{user},\att{event},\att{\#likes})$, and
$\rel{Attend}(\att{user},\att{event},\att{\#likes})$. 
We wish to extract 
statistics about triples of users involved in events, beginning by joining the three relations
using the \JQ
\[
    \rel{Admin}(u_1,e)\,,\,\rel{Share}(u_2,e,l_2)\,,\,\rel{Attend}(u_3,e,l_3).
\]
Now, suppose that all we do is apply a quantile query to the result of $Q$, say the $0.1$-quantile ordered by
$l_2+l_3$ (the sum of likes of the share and the participation).
The direct way of finding the quantile is to materialize the join, sort the resulting tuples, and take the element at position $\quant = 0.1$. Yet, this result might be considerably larger than the database, and prohibitively expensive to compute, even though in the end we care only about one value. Can we do better?  This is the research question that we study in this paper. In general, the answer depends on the \JQ and order, and in this specific example we can, actually, do considerably better.

To be more precise, we study the fine-grained data complexity of query evaluation, where the query seeks a quantile over a \JQ. We refer to such a query as a \e{Quantile
Join Query} and abbreviate it as \%\JQ. So, the \JQ $Q$ is fixed, and so is the ranking function (e.g., weighted sum over a fixed subset of the variables). 
The input consists of $D$ and $\quant$. 
In terms of the execution cost, we allow for poly-logarithmic factors, therefore our goal is to devise evaluation algorithms that run in time quasilinear in $D$, that is, $\bigO(|D|\polylog(|D|))$. 

To the best of our knowledge, little is known about the fine-grained complexity of \%\JQ. 
We have previously studied this problem~\cite{carmeli23direct} for Conjunctive Queries
(which are more general than Join Queries since they also allow for projection)
under the name
``\e{selection problem}''.\footnote{To be precise, in the selection problem, the position of this answer is given as an absolute index $i$ rather than a relative position $\quant$. (This problem is sometimes referred to as \e{unranking}~\cite{DBLP:journals/ipl/MyrvoldR01}.)
This difference is nonessential  
as far as this work concerns: 
all our previous lower and upper bounds for selection on \JQs 
apply to \%\JQs.} 
Two types of orders were covered: sum of all attributes and lexicographic orders.
On the face of it, the conclusion from our previous results is that we are extremely limited in what we can do: \e{The problem is intractable for every \JQ with more than two atoms} (each having a set of variables that is not contained in that of another atom, and assuming no self-joins), under conventional conjectures in fine-grained complexity. 
Nevertheless, we argue that our previous results tell only part of the story and miss quite general opportunities for tractability:
\begin{itemize}
    \item What if the sum involves just a \e{subset} of the variables, like in the above social-network example? Then the lower bounds for full sum do not apply. As it turns out, in this case we often can achieve tractability for \JQs of more than 2 atoms.
    \item What if we allow for some small error and not insist on the precise $\quant$-quantile? As we argue later, this relaxation makes the picture dramatically more positive.
\end{itemize}
In addition, there are ranking functions that have not been considered at all, notably minimum and maximum over attributes, such as $\rankmin(\att{rate1},\att{rate2})$ and 
$\rankmax(\att{width},\att{height},\att{depth})$. We do not see any conclusion from past results on these, so the state of affairs (prior to this paper) is that their complexity is an open problem.

In this work, we devise a new framework for evaluating \%\JQ queries. We view the problem as a search problem in the space of query answers, and
 the framework adopts a divide-and-conquer ``pivoting'' approach. 
For a \JQ $Q$, we reduce the problem to two subroutines, given $D$ and $\quant$:
\begin{itemize}
    \item \textbf{\PIVOT}: Find a \e{pivot} answer $p$ such that the set of answers that precede $p$ and the set of answers that follow $p$ both contain at least a constant fraction of the answers.
    \item \textbf{\TRIM}: Partition the answers into three splits: less than, equal to, and greater than $p$.
    Determine which one
    contains the sought answer 
    and, if it is not $p$, produce a new \%\JQ within the relevant split using new $Q'$, $D'$ and $\quant'$.
    We view this operation as \e{trimming} the split condition by updating the database so that the remaining answers satisfy the condition. 
\end{itemize}

We begin by showing that we can select a pivot in linear time for \e{every} acyclic join query and \e{every} ranking function that satisfies 
a monotonicity assumption (also used in the problem of ranked enumeration~\cite{tziavelis22anyk,KimelfeldS2006}),
which all functions considered here satisfy. 
Note that the assumption of acyclicity is required since, otherwise, it is impossible to even determine whether the join query has any answer in quasilinear time, under conventional conjectures in fine-grained complexity~\cite{bb:thesis}. Hence, the challenge really lies in trimming. 

\subsubsection*{Contributions}
Using our approach, we establish efficient algorithms for several classes of queries and ranking functions, where we show how to solve the trimming problem.
\begin{enumerate}
  \item We establish tractability for all acyclic \JQs under the ranking functions $\rankmin$ and $\rankmax$.
  \item %
  We recover (up to logarithmic factors) all past tractable cases~\cite{carmeli23direct} for lexicographic orders and $\ranksum$.
  \item %
  For self-join-free \JQs and \ranksum, we complete the picture
  by extending the previous dichotomy~\cite{carmeli23direct}
  (restricted to \JQs)
  to all \e{partial sums}.
\end{enumerate}

We then turn our attention to approximate answers. Precisely, we find an answer at a position within
$(1\pm\epsilon)\quant$ for an allowed error $\epsilon$.
(This is a standard notion of approximation for quantiles~\cite{DBLP:conf/sigmod/RajagopalanML98,DBLP:conf/icdt/DoleschalBKM21}.)
To obtain an efficient \e{randomized} approximation, it suffices to be able to construct in quasilinear time a direct-access structure for the underlying \JQ, regardless of the answer ordering; if so, then one can use a standard median-of-samples approach (with Hoeffding's inequality to guarantee the error bounds). 
Such algorithms for direct-access structures have been established in the past for arbitrary acyclic \JQs~\cite{oldRandomAccess,bb:thesis}.
Instead, we take on the challenge of \e{deterministic} approximation.
Our final contribution is that:
\begin{enumerate}
\setcounter{enumi}{3}
\item We show that with an adjustment of our pivoting framework, 
we can establish a deterministic approximation scheme in time
quadratic in $1/\epsilon$ and quasilinear in database size. 
\end{enumerate}
In contrast to the randomized case, we found the task of deterministic approximation challenging, and our algorithm is indeed quite involved.
Again, the main challenge is in the trimming phase.

The remainder of the paper is organized as follows: 
We give preliminary definitions in \cref{sec:preliminaries}. 
We describe the general pivoting framework in \cref{sec:quantile-alg}. 
In \cref{sec:pivot}, we describe the pivot-selection algorithm. 
The main results are in \cref{sec:exact,sec:approx} where we devise exact and approximate trimmings, respectively, and establish the corresponding tractability results. 
We conclude in \cref{sec:conclusions}.

\section{Preliminaries}\label{sec:preliminaries}

\subsection{Basic Notions}

\introparagraph{Sets}
We use 
$[r]$ to denote the set of integers $\{1, \ldots, r\}$.
A multiset $L$ is described by a 2-tuple $(Z, \beta)$, where
$Z$ is the set of its distinct elements 
and $\beta : Z \rightarrow \N$ is a multiplicity function.
The set of all possible multisets with elements $Z$ is denoted by $\N^{Z}$.

\introparagraph{Relational databases}
A \e{schema} $\calS$ is a set of relational symbols $\{ R_1, \ldots, R_m \}$.
A \e{database} $D$ 
contains a finite relation $R^D \subseteq \dom^{a_R}$ for each $R \in \calS$, 
where $\dom$ is a set of constants called the \e{domain},
and $a_R$ is the arity of symbol $R$.
If $D$ is clear, we simply use $R$ instead of $R^D$.
The size of $D$ is the total number of tuples, denoted by $n$.

\introparagraph{Join Queries}
A \e{Join Query} (\JQ) $Q$ over schema $\calS$ 
is an expression of the form
$R_1(\mathbf{X}_1), \ldots, R_\ell(\mathbf{X}_\ell)$,
where $\{ R_1, \ldots, R_\ell \} \subseteq \calS$
and the variables of $Q$ are
$\var(Q) = \cup_{i \in [\ell]} \mathbf{X}_i$,
sometimes interpreted as a tuple instead of a set.
Each $R_i(\mathbf{X}_i), i \in [\ell]$ is called an \e{atom} of $Q$.
A repeated occurrence of a relational symbol is a \e{self-join} and a
\JQ without self-joins is \e{self-join-free}.
A \e{query answer} is a homomorphism 
from $Q$ to the database $D$, i.e.
a mapping from $\var(Q)$ to $\dom$ constants,
such that every atom $R_i(\mathbf{X}_i), i \in [\ell]$ maps to a tuple of $R_i^D$.
The set of query answers to $Q$ over $D$ is denoted by $Q(D)$ and we often represent a query answer $\qanswer \in Q(D)$ 
as a tuple of values assigned to $\var(Q)$.
For an atom $R_i(\mathbf{X}_i)$ of a \JQ and database $D$, 
we say that tuple $t \in R_i^D$ assigns value $a$ to variable $x$, and write it as $t[x] = a$,
if for every index $j$ such that $\mathbf{X}_i[j] = x$ we have $t[j] = a$.
For a predicate $P(\mathbf{X}_P)$ over variables $\mathbf{X}_P \subseteq \var(Q)$,
we denote by $(Q \wedge P)(D)$ the subset of query answers $Q(D)$ that satisfy $P(\mathbf{X}_P)$.

\introparagraph{Hypergraphs}
A {\em hypergraph} $\calH=(V,E)$ is a set $V$ of {\em vertices} and a set $E \subseteq 2^V$ of {\em hyperedges}.
A {\em path} in $\calH$ is a sequence of vertices such that every two consecutive vertices appear together in a hyperedge.
A {\em chordless path} is a path in which no two non-consecutive ones appear in the same hyperedge 
(in particular, no vertex appears twice).
The \e{number of maximal hyperedges}
is
$\mh(\calH ) = | \{ e \in E \mid \nexists e' \in E: e \subset e' \}|$.
A set of vertices $U \subseteq V$ 
is
\e{independent} if no pair appears in a hyperedge, i.e.,
$|U \cap e| \leq 1, \forall e \in E$.

\introparagraph{Join trees}
A \emph{join tree} of a hypergraph $\calH=(V,E)$ is a tree $T$ where its nodes\footnote{
To avoid confusion, we use the terms hypergraph \e{vertices} and tree \e{nodes}.}
are the hyperedges of $\calH$
and the {\em running intersection} property holds, namely: 
for all $u \in V$ the set $\{e \in E \mid u \in e\}$ forms a (connected) subtree in $T$.
We associate a hypergraph $\calH(Q) = (V, E)$ to a \JQ $Q$ 
where the vertices are the variables of $Q$, 
and every atom of $Q$ corresponds to a hyperedge with the same set of variables.
With a slight abuse of notation, we identify atoms of $Q$ with hyperedges of $\calH(Q)$.
A \JQ $Q$ is {\em acyclic} if there exists a join tree for $\calH(Q)$,
otherwise it is \e{cyclic}.
If we root the join tree, the subtree rooted at a node $U$ defines a subquery,
i.e., a \JQ that contains only the
atoms of descendants of $U$.
A partial query answer (for the subtree) rooted at $U$
is an answer to the subquery.
If we materialize a relation $R_U$ for node $U$,
a partial query answer (for the subtree) rooted at $t \in R_U$
must additionally agree with $t$.

\introparagraph{Complexity}
We measure complexity in the database size $n$,
while query size is considered constant.
The model of computation is the standard RAM model with uniform cost measure.
In particular, it allows for linear-time construction of lookup tables,
which can be accessed in 
constant time.
Following our prior work~\cite{carmeli23direct}, we only consider comparison-based sorting,
which takes quasilinear time.

\subsection{Orders over Query Answers}
To define \%\JQs,
we assume an ordering of the query answers by a given ranking function.
The ranking function is described by a 2-tuple $(\w, \preceq)$
where a weight function $\w : Q(D) \rightarrow \Wdom$ maps the answers
to a weight domain $\Wdom$ 
equipped with a total order $\preceq$.
We denote the strict version of the total order by $\prec$.
Assuming consistent tie-breaking,
the total order extends to query answers,
i.e., for $q_1, q_2 \in Q(D)$, $q_1 \preceq q_2$ iff $\w(q_1) \prec \w(q_2)$
or $\w(q_1) = \w(q_2)$ and $q_1$ is (arbitrarily but consistently) chosen to break the tie.

\introparagraph{Weight aggregation model}
We focus on the case of \emph{aggregate ranking functions} where 
the query answer weights are computed by aggregating
weights are assigned to
the input database.
In particular, an input-weight function $\w_x: \dom \rightarrow \Wdom$ 
associates each value of variable $x$ with a weight in $\Wdom$.
An aggregate function $\Waggr: \N^{\Wdom} \rightarrow \Wdom$ 
takes a multiset of weights and produces a single weight.
Aggregate ranking functions are typically not sensitive to the order in which the input weights are given~\cite{gray97cube,jesus15aggregation},
captured by the fact that their input is a multiset.
Query answers map to $\Wdom$ by aggregating the weights of values assigned to a subset of the input variables $\Wvars \subseteq \var(Q)$
with an aggregate function $\Waggr$.
Thus, the weight of a query answer $\qanswer \in Q(D)$ is
$\w(\qanswer) = \Waggr(\{\w_x(\qanswer[x]) \mid x \in \Wvars \})$.
When we do not have a specific assignment from variables to values,
we use $\w(\Wvars)$ to refer to the expression $\Waggr(\{\w_x(x) \mid x \in \Wvars \})$.
For example, if $\var(Q) = \{x_1, x_2, x_3\}, \Wvars = \{x_1, x_3\}, \w_x(x)$ is the identity function 
for all varaibles $x$, and $\Waggr$ is summation,
then $\w(\Wvars) = x_1 + x_3$.

\introparagraph{Concrete ranking functions}
In this paper, we discuss three types of ranking functions:

\begin{enumerate}
	\item \textbf{\ranksum}: 
        $\Wdom$ is $\R$ and $\Waggr$ is summation. We use the term \e{full \ranksum} when $\Wvars=\var(Q)$ and \e{partial \ranksum} otherwise.
 
	\item \textbf{\rankmin / \rankmax}:
         $\Wdom$ is $\R$ and $\Waggr$ is $\min$ or $\max$.

    \item \textbf{\ranklex}: Lexicographic orders fit into our framework by letting the domain $\Wdom$ consist of tuples in
    $\N^{|\Wvars|}$. 
    Every variable $x \in \Wvars$ is mapped to $\Wdom$ as a tuple
    $(0, \ldots, w_x'(x), \ldots, 0)$ where 
    $w_x'(x)$ occupies
    the position of $x$ in the lexicographic order
    and $w_x'$ 
    is a function $w_x' : \dom \rightarrow \N$ that orders the domain of $x$ by mapping it to natural numbers.
    The aggregate function $\Waggr$ is then element-wise addition, while the
    order $\preceq$ compares these tuples lexicographically.
\end{enumerate}

\introparagraph{Problem definition}
Let $Q$ be a \JQ and $(\w, \preceq)$ a ranking function.
Given a database $D$, a query answer $q \in Q(D)$ is a $\quant$-quantile~\cite{rice2007msd} of $Q(D)$ for some $\quant \in [0,1]$
if there exists a valid ordering of $Q(D)$
where there are
$\lceil \quant |Q(D)| \rceil$ answers less-than or equal-to
$q$ and $\lfloor (1 - \quant) |Q(D)| \rfloor$
answers greater than $q$.
A \%\JQ asks for a $\quant$-quantile given $D$ and $\quant$.
Similarly, an $\epsilon$-approximate \%\JQ asks
for a $(\quant \pm \epsilon)$-quantile for a given $D$, $\quant$,
and $\epsilon \in (0, 1)$.

\introparagraph{Monotonicity}
Let $\uplus$ be multiset union.
An (aggregate) ranking function is \e{subset-monotone}~\cite{tziavelis22anyk} if
$\Waggr(L_1) \preceq \Waggr(L_2)$ implies that $\Waggr(L \uplus L_1) \preceq \Waggr(L \uplus L_2)$
for all multisets $L, L_1, L_2$.
All ranking functions we consider in this work have this property.
We note that subset-monotonicity has been used as an assumption in ranked enumeration~\cite{tziavelis22anyk,KimelfeldS2006} and
is a stronger requirement than the more well-known monotonicity 
notion of Fagin et al.~\cite{fagin03}.

\introparagraph{Tuple weights}
Our ranking function definition uses 
attribute-weights
but some of our algorithms
are easier to describe when dealing with tuple weights.
We can convert the former to the latter in linear time.
First, we eliminate self-joins by materializing a fresh relation for every repeated symbol in the query $Q$.
Second, to avoid giving the weight of a variable to tuples of multiple relations,
we define a mapping $\mu$ that assigns each variable $x \in \Wvars$ 
to a relation $R$ such that $x$ occurs in the $R$-atom of $Q$.
The weight of a tuple $t \in R$ is then the multiset of weights for variables assigned to $R$:
$\w_R(t) = \{\w_x(t[x]) \mid x \in \Wvars, \mu(x) = R\}$.
\footnote{The reason that 
we maintain the attribute weights as a set instead of aggregating them is that the aggregate 
ranking function can be holistic~\cite{gray97cube}, in which case we lose the ability to further aggregate. 
If, on the other hand, the ranking function is distributive~\cite{gray97cube} like \ranksum,
then we can aggregate to obtain a single weight for a tuple.
}
The total order $\preceq$ can be extended to sets of tuples $(t_1, \ldots, t_r)$
(and thus query answers)
by aggregating all individual weights contained in the tuple weights.

\subsection{Known Bounds}
\label{sec:known_bounds}

Certain upper and (conditional) lower bounds for  
\%JQ follow from our previous work~\cite{carmeli23direct}
on the \e{selection problem}
which asks for the query answer at index $k$.
The two problems are equivalent for acyclic \JQs, since an index 
can be translated into a fraction $\quant$,
and vice-versa,
by knowing $|Q(D)|$,
which can be computed in linear time as we explain in \Cref{sec:message_passing}.

The lower bounds are based on two hypotheses:

\begin{enumerate}
\item
\hyperclique~\cite{abboud14conjectures,DBLP:conf/soda/LincolnWW18}:
Let a \e{$(k{+}1,k)$-hyperclique} be a set of $k{+}1$ vertices
such that every subset of $k$ elements is a hyperedge.
For every $k \geq 2$, there is no
$O(m \polylog m)$ algorithm for deciding the existence of a
$(k{+}1,k)$-hyperclique in a $k$-uniform hypergraph with $m$ hyperedges.

\item
\ThreeSUM~\cite{patrascu2010dynamic,ilya053sum}:
For any $\epsilon > 0$,
we cannot decide in time $O(m^{2-\epsilon})$ 
whether there exist $a \in A, b \in B, c \in C$
from three integer sets $A, B, C$, each of size $\Omega(m)$, such that 
$a + b + c = 0$.

\end{enumerate}
\hyperclique{} implies that
we cannot decide
in $\bigO(n \polylog n)$ if a cyclic, self-join-free \JQ has any answer~\cite{bb:thesis}.
For \ranklex, an acyclic \%\JQ can be answered in $\bigO(n)$~\cite{carmeli23direct}.
For full \ranksum, an acyclic \%\JQ can be answered in $\bigO(n \log n)$ if its maximal hyperedges are at most 2,
and the converse is true if it is also self-join-free,
assuming \ThreeSUM~\cite{carmeli23direct}.

\subsection{Message Passing}
\label{sec:message_passing}

Message passing is a common algorithmic pattern that many algorithms for acyclic \JQs follow.
For example, it allows us to count the number of answers to an acyclic \JQ in linear time~\cite{khamis16faq,pichler13counting}.
Some of the algorithms that we develop also follow this pattern, that we abstractly describe below.

\introparagraph{Preprocessing}
Choose an arbitrary root for a join tree $T$ of the \JQ,
and materialize a distinct relation for every $T$-node.
For every parent node $V_p$ and child node $V_c$, group the $V_c$-relation
by the $V_p \cap V_c$ variables.
We will refer to these groups of tuples as \e{join groups}; a join group shares the same 
values for variables that appear in the parent node.
The algorithm visits the relations in a bottom-up order of $T$,
sending children-to-parent messages.
The goal is to compute a value $\mval(t)$ for each tuple $t$ of these relations,
initialized according to the specific algorithm.
Sometimes, it is convenient to add
an artificial root node $V_0 = \emptyset$ to the join tree,
which refers to a zero-arity relation with a single tuple $t_0 = ()$.
Tuple $t_0$ joins with all tuples of the previous root
and its purpose is only to gather the final result at the end of the bottom-up pass.

\introparagraph{Messages}
As we traverse the relations in bottom-up order, every tuple $t$ emits its $\mval(t)$.
These messages are aggregated as follows:

\begin{enumerate}
    \item Messages emitted by tuples $t'$ in a join group are aggregated with an operator $\oplus$.
    The result is sent to all parent-relation tuples that agree with the join values of the group. 
    
    \item A tuple $t$ computes $\mval(t)$ by aggregating the messages received from all children in the join tree, together with the initial value of $\mval(t)$, with an operator $\otimes$.
    
\end{enumerate}

\begin{figure}[t]
\centering
\begin{subfigure}{.45\linewidth}
    \centering
    \includegraphics[scale=0.48]{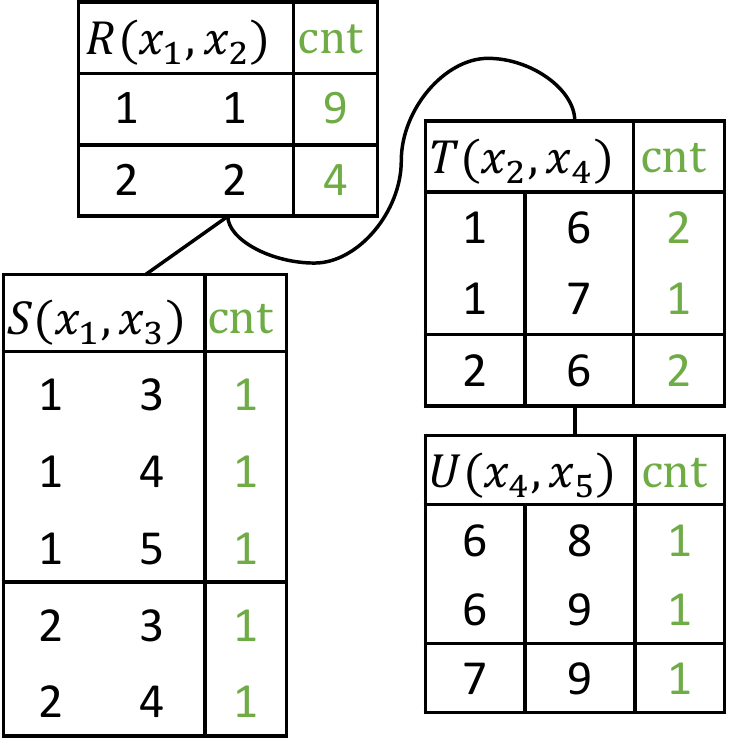}
    \caption{Example database and final counts of subtree answers.}
    \label{fig:tables}
\end{subfigure}%
\hfill
\begin{subfigure}{.53\linewidth}
    \centering
    \includegraphics[scale=0.49]{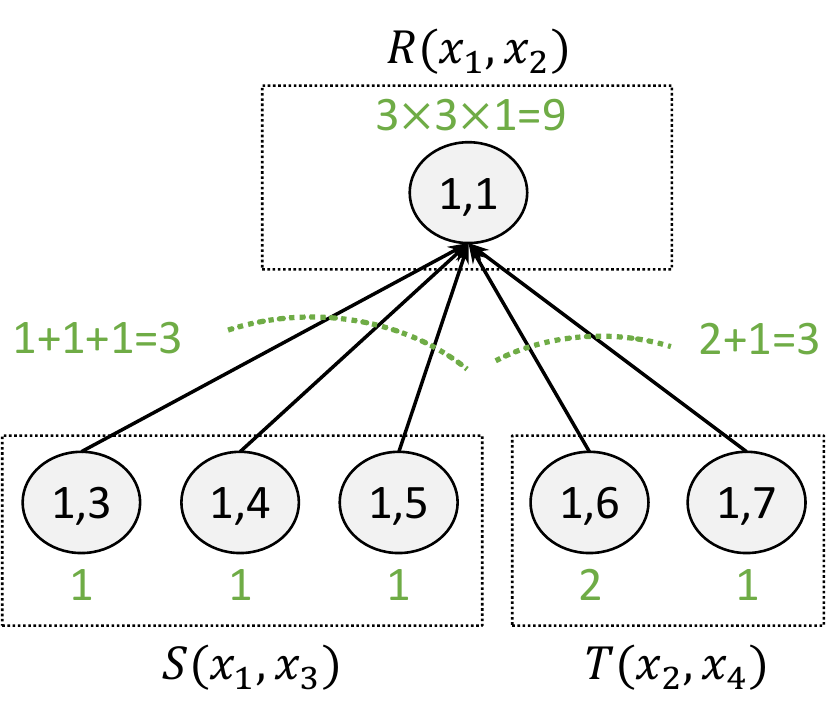}
    \caption{Messages received and aggregated by an $R$-tuple.}
    \label{fig:pivot-computatuon}
\end{subfigure}
\caption{Message passing for counting the answers to the \JQ
$R(x_1,x_2), S(x_1, x_3), T(x_2,x_4),U(x_4, x_5)$.}
\label{fig:count}
\end{figure}

\begin{example}[Count]
To count the \JQ answers,
we initialize $\cnt(t)=1$ for all tuples $t$,
set $\oplus$ to product $(\prod)$,
and $\otimes$ to sum $(\sum)$.
\Cref{fig:count} illustrates how messages are aggregated so that $\cnt(t)$ is 
the number of partial answers for the subtree rooted at $t$.
To get the final count, we sum the counts in the root relation ($9+4=13$ in the example),
e.g., by
introducing the artificial root-node tuple $t_0$.
\end{example}

\section{Divide-and-Conquer Framework}\label{sec:quantile-alg}

We describe a general divide-and-conquer framework for acyclic \%\JQs that applies to different ranking functions.
It follows roughly the same structure as linear-time
selection~\cite{blum73select} in a given array of elements.
This classic algorithm searches for the element at a desired index $k$
in the array by ``pivoting''.
In every iteration, it selects a pivot element and creates three array partitions: 
elements that are lower, equal, and higher than the pivot.
Depending on the partition sizes and the value of $k$,
it chooses one partition and continues with that, 
thereby reducing the number of candidate elements.
We adapt the high-level steps of this algorithm to our setting.
The crucial challenge is that 
\emph{we do not have access to the materialized array} of query answers
(which can be very large),
but only to the input database and \JQ that produce them.
In the following, we discuss the general structure of the algorithm
and the subroutines that are required for it to work.
In later sections, we then concretely specify these subroutines.

\introparagraph{Pivot selection}
We define what constitutes a ``good'' pivot.
Intuitively, it is an element whose position is roughly in the middle of the 
ordering. 
With such a pivot, the partitioning step is guaranteed to
eliminate a 
significant number 
of elements,
resulting in quick convergence.
Ideally, we would want to have the true median as our pivot because it 
is guaranteed to eliminate the largest fraction ($\frac{1}{2}$) of elements.
However, to achieve convergence in a logarithmic number of iterations,
it is sufficient to choose any pivot that eliminates any
constant fraction $c > 0$ of elements.

\begin{definition}[$c$-pivot]\label{def:c-pivot}
For a constant $c \in (0, 1)$ 
and a set $Z$ equipped with a total order $\preceq$, 
a $c$-pivot $p$ for $Z$ is an element of $Z$ such that
$|\{ z \in Z \mid p \preceq z \}| \geq c|Z|$ and
$|\{ z \in Z \mid p \succeq z \}| \geq c|Z|$.
\end{definition}

Our goal is to find such a $c$-pivot for the set of query answers $Q(D)$ ordered by the given ranking function.

\introparagraph{Partitioning}
Assuming an appropriate query answer $p$ as our pivot, we use it to partition
the query answers.
This means that we want to separate the answers into those
whose weight is
less than, equal to, and greater than the weight of the pivot.
Since we do not have access to the query answers, this partitioning step must be performed on the input database and \JQ.
The less-than and greater-than partitions can be described by the original \JQ, together with
inequality predicates:
(1) $\w(\Wvars) \prec \w(p)$ and
(2) $\w(\Wvars) \succ \w(p)$ respectively.
The equal-to partition can be assumed to contain all answers that do not fall into either of the other two.

\introparagraph{Trimming inequalities}
If we materialize as
database relations the inequalities that arise from the partitioning step, their size can be very large.
For example, the inequality $x_1 + x_2 + x_3 < 0$
for three variables $x_1, x_2, x_3$
has a listing representation of size $\bigO(n^3)$.
However, in certain cases it is possible to represent them
more efficiently, e.g., in space $\bigO(n \polylog n)$,
by modifying the original \JQ and database.
We call this process ``\e{trimming}.''

\begin{definition}[Predicate Trimming]
\label{def:exact_trimming}
Given a \JQ $Q$ and a predicate $P(U)$ with variables $U \subseteq \var(Q)$,
a \emph{trimming} of $P(U)$ from $Q$
receives a database $D$ and
returns a \JQ $Q'$ of size $\bigO(|Q|)$
and with $\var(Q) \subseteq \var(Q')$,
and a database $D'$
for which there exists an $\bigO(1)$-computable bijection from $Q'(D')$
to $(Q \wedge P)(D)$.
Trimming time is the time required to construct $Q'$ and $D'$.
\end{definition}

Efficient trimmings of predicates are for instance known for additive inequalities when the sum variables are found in adjacent \JQ atoms~\cite{tziavelis21inequalities}
and for not-all-equals predicates~\cite{khamis19negation}, which are a generalization of non-equality ($\neq$).
Ultimately, our ability to partition and the success of our approach relies on the existence of efficient trimmings
of inequalities that involve the aggregate function.

\introparagraph{Choosing a partition}
After we obtain three new \JQs and corresponding databases
by trimming,
we count their query answers 
to determine where the desired index
(calculated from the given percentage)
falls into.
To ensure that this can be done in linear time,
we want all \JQs to be acyclic,
and so we restrict ourselves to trimmings that do not alter the acyclicity of the \JQs.
To keep track of the candidate query answers, 
we maintain two weights $\texttt{low}$ and $\texttt{high}$ as bounds,
which define a contiguous region in the sorted array of query answers.
Every iteration then applies trimming for two additional inequalities
$\w(\Wvars) \succ \texttt{low}$ and $\w(\Wvars) \prec \texttt{high}$
in order to restrict the search to the current candidate set.

\introparagraph{Termination}
The algorithm terminates when the desired index falls into the equal partition since any of its answers, including our pivot,
is a $\quant$-quantile.\footnote{
If we want to enforce the same tie-breaking scheme 
across different calls to our algorithm,
we could continue searching within the equal partition with a \ranklex order, but this requires also trimming for equality-type predicates.
}
It also terminates when the number of candidate answers
is sufficiently small,
by calling
the Yannakakis algorithm~\cite{Yannakakis} to materialize them
and
then applying
linear-time selection~\cite{blum73select}.
With $c$-pivots, we eliminate at least $c |Q(D)|$
answers in every iteration;
hence, the candidate query answers 
will be $\bigO(n)$ after a logarithmic number of iterations.
Notice that our trimming definition allows the new database
to be larger than the starting one,
so the database size may increase across iterations.
However, the number of \JQ answers decreases, ensuring termination.

\introparagraph{Summary}
Our algorithm repeats the above steps
(selecting pivot, partitioning, trimming)
iteratively.
It requires the implementation
of two subroutines:
(1) selecting a $c$-pivot among the \JQ answers, which we call ``\PIVOT'',
and (2) trimming of inequalities,
which we call ``\TRIM''.
The pseudocode is in \Cref{sec:details_quantile}.

\begin{lemma}[Exact Quantiles]
\label{lem:quantile_exact}
Let $\mathcal{Q}$ be a class of acyclic \JQs and $(\w, \preceq)$ a ranking function.
If for all $Q' \in \mathcal{Q}$ 
\begin{enumerate}
    \item there exists a constant $c$ such that for any database $D$,
    a $c$-pivot for $Q'(D)$ can be computed in time $\bigO(g_p(n))$
    for some function $g_p$, and
    \item for all $\wval \in \Wdom$, there exist trimmings of $\w(\Wvars) \prec \wval$ and $\w(\Wvars) \succ \wval$ from $Q'$ that return $Q'' \in \mathcal{Q}$ in time $\bigO(g_t(n))$ for some function $g_t$,
\end{enumerate}
then a \%\JQ can be answered for all $Q \in \mathcal{Q}$ in time
$\bigO(\max \{ g_p(n), g_t(n) \} \log n)$.
\end{lemma}

Notice that trimming can result in a \e{different} query
than the one we start with.
This is why pivot-selection and trimming need to applicable not just to the input query $Q$,
but to all queries that we may obtain from trimming (referred to as a class in \Cref{lem:quantile_exact}).

\begin{example}
Suppose that $Q$ is $R_1(x_1, x_2), R_2(x_2, x_3)$ over a database $D$
and we want to compute the median by \ranksum with attribute weights equal to their values.
First, we call \PIVOT to obtain a 
pivot answer $p$,
which we use to create two partitions:
one with $x_1 + x_2 + x_3 < \w(p)$,
and one with $x_1 + x_2 + x_3 > \w(p)$.
Second, we call \TRIM on these inequalities.
By a known construction~\cite{tziavelis21inequalities}, these inequalities
can be trimmed in $\bigO(n \log n)$.
This construction 
adds 
a new column and variable $v$ to both relations.
We now have two \JQs $Q_\texttt{lt}$ and
$Q_\texttt{gt}$,
over databases
$D_\texttt{lt}$,
$D_\texttt{gt}$.
For example, $Q_\texttt{lt}$ is
$R_{1_\texttt{lt}}(x_1, x_2, v), R_{2_\texttt{lt}}(v, x_2, x_3)$.
Suppose that $|Q(D)| = 1001$, hence the desired index is $k = 500$ (with zero-indexing).
If $|Q_\texttt{lt}(D_\texttt{lt})| = 400$
and $|Q_\texttt{gt}(D_\texttt{gt})| = 600$,
we can infer that the middle partition contains $1$ one answer with weight $\w(p)$.
Thus, we have to continue searching in the index range from $401$ to $1000$ with $k' = 500 - 400 - 1 = 99$.
To create less-than and greater-than partitions in the next iteration, we will 
start with the original $Q$ and $D$ and
apply inequalities
$\w(p) \leq x_1 + x_2 + x_3 < \w(p')$
and
$\w(p') < x_1 + x_2 + x_3 < \infty$
 with some new pivot $p'$.
\end{example}

In \Cref{sec:pivot}, we will show that an efficient algorithm for \PIVOT exists 
for \emph{any subset-monotone ranking function}.
For \TRIM, the situation is more tricky and no generic solution is known.
For each ranking function, we design a trimming algorithm tailored to it.
This is precisely where we encounter the known conditional hardness of \ranksum~\cite{carmeli23direct}.
For example, a quasilinear trimming for $Q(x_1, x_2, x_3) \datarule R_1(x_1), R_2(x_2), R_3(x_3)$ and $x_1 + x_2 + x_3 < 0$ would 
violate our \ThreeSUM{} hypothesis
(see \Cref{sec:known_bounds})
because it would allow us to
count the number of answers in the less-than and greater-than partitions.
In \Cref{sec:exact}, we will show that efficient trimmings exist for
\rankmin/\rankmax and \ranklex,
as well as (partial) \ranksum in certain cases.

\subsection{Adaptation for Approximate Quantiles}
Since 
\%\JQ can be intractable (under our efficiency yardstick) for some ranking functions such as \ranksum~\cite{carmeli23direct},
we aim for $\epsilon$-approximate quantiles.
We can obtain a \e{randomized} approximation 
by the standard technique of sampling answers uniformly and taking as estimate the $\quant$-quantile of the sample (e.g., as done by Doleschal et al.~\cite{DBLP:conf/icdt/DoleschalBKM21} for quantile queries in a different model). Concentration theorems such as Hoeffding Inequality imply that it suffices to collect $O(1/\epsilon^{2})$ samples and repeat the process
$O(\log(1/\delta))$ times (and select the median of the estimates) to get an $\epsilon$-approximation with probability $1-\delta$. So, it suffices to be able to efficiently sample uniformly a random answer of an acyclic \JQ; we can do so using linear-time algorithms for constructing a logarithmic-time random-access structure for the answers~\cite{oldRandomAccess,bb:thesis}.

We will show that with our pivoting approach, we can obtain a \e{deterministic} approximation, which we found to be much more challenging than the randomized one. 
Pivot selection remains the same as in the exact algorithm,
while for trimming  (which as we explained is the missing piece for \ranksum),
we introduce an approximate solution based on the notion of a \emph{lossy trimming}.
Intuitively, a lossy trimming does not retain all the \JQ answers
that satisfy a given predicate.
Such a trimming results in some valid query answers being lost in each iteration
and causes inaccuracy in the final index of the returned query answer.
However, if the number of lost query answers is bounded, then we can also bound
the error on the index.

\begin{definition}[Lossy Predicate Trimming]
Given a \JQ $Q$, a predicate $P(U)$ with variables $U \subseteq \var(Q)$, and
a constant $\epsilon \in [0, 1)$,
an \emph{$\epsilon$-lossy trimming} of $P(U)$ from $Q$
receives a database $D$ and
returns a \JQ $Q'$ of size $\bigO(|Q|)$
and with $\var(Q) \subseteq \var(Q')$,
and a database $D'$
for which 
there exists an $\bigO(1)$-computable injection from $Q'(D')$ to $(Q \wedge P)(D)$, 
and also
$|Q'(D')| \geq (1 - \epsilon) |(Q \wedge P)(D)|$.
Trimming time is the time required to construct $Q'$ and $D'$.
\end{definition}

The injection in the definition above implies that some query answers that satisfy the given predicate 
do not correspond to any answers in the new instance we construct,
but we also ask their ratio to be bounded by $\epsilon$.
For $\epsilon=0$, we obtain an exact predicate trimming (\Cref{def:exact_trimming})
as a special case.\footnote{
An \emph{imprecise} trimming, which retains \emph{more} \JQ answers
than it should, would also work for our quantile algorithm.}

\begin{lemma}[Approximate Quantiles]
\label{lem:quantile_approx}
Let $\mathcal{Q}$ be a class of acyclic \JQs and $(\w, \preceq)$ a ranking function.
If for all $Q' \in \mathcal{Q}$ 
\begin{enumerate}
    \item there exists a constant $c$ such that for any database $D$,
    a $c$-pivot for $Q'(D)$ can be computed in time $\bigO(g_p(n))$
    for some function $g_p$, and

    \item for all $\wval \in \Wdom$ and $\epsilon' > 0$, there exist $\epsilon'$-lossy trimmings 
    of $\w(\Wvars) \prec \wval$ and $\w(\Wvars) \succ \wval$
    from $Q'$ that return $Q'' \in \mathcal{Q}$ in time $\bigO(g_t(n, \epsilon'))$
    for some function $g_t$,
\end{enumerate}
then an $\epsilon$-approximate
\%\JQ can be answered for all $Q \in \mathcal{Q}$ in time
$\bigO \big(\max \bigl\{ g_p(n), g_t \big(n, \frac{\epsilon}{2 \lceil \ell \log_{1/(1-c)} n \rceil} \big) \bigr\} \log n \big)$, 
where $\ell$ is the number of atoms of $Q$.
\end{lemma}

In \Cref{sec:approx} we will give an $\epsilon$-lossy trimming for additive inequalities,
which, combined with the pivot algorithm of \Cref{sec:pivot}, will give us an $\epsilon$-approximate quantile algorithm for \ranksum.

\section{Generic Pivot Selection}
\label{sec:pivot}

We describe a \PIVOT algorithm for choosing a pivot element among the answers to an acyclic \JQ. 
This is one of the two main subroutines of our quantile algorithm.
We show that a $c$-pivot can be computed in linear time for a large class of ranking functions.

\begin{lemma}[Pivot Selection]
\label{lem:pivot}
Given an acyclic \JQ $Q$ over a database $D$ of size $n$ and a \rankf ranking function,
a $c$-pivot of $Q(D)$ together with $c \in (0, 1)$
can be computed in time $\bigO(n)$. 
\end{lemma}

\subsection{Algorithm}

The key idea of our algorithm
is the ``median-of-medians'', in similar spirit to classic linear-time selection~\cite{blum73select} or selection for the $X+Y$ problem~\cite{johnson78xy}.
The main difference is that we apply the median-of-medians idea
iteratively using message passing.
The detailed pseudocode is given in \Cref{sec:pivot_details}.

\introparagraph{Weighted median}
An important operation for our algorithm is the weighted median,
which selects the median of a set,
assuming that every element appears a number of times equal to an assigned weight.\footnote{This weight is not the same as the weight assigned by the ranking function, thus we simply call it multiplicity.}
More formally, for a totally-ordered ($\preceq$) set $Z$ and a multiplicity function $\beta : Z \rightarrow \N^+$,
the weighted median $\wm_\preceq(Z, \beta)$ is the element at position $\lfloor \frac{|B|}{2} \rfloor$
in the multiset $B = (Z, \beta)$ ordered by $\preceq$.
The weighted median can be computed in time linear in $|Z|$~\cite{johnson78xy}.

\introparagraph{Message passing}
Our algorithm employs the message-passing framework as outlined in \Cref{sec:message_passing}
to compute $\pivot(t)$ for each tuple $t$ bottom-up.
The computed $\pivot(t)$ is a partial query answer for the subtree rooted at $t$
and serves as a $c'$-pivot for these partial answers,
for some $c' \geq c$.
Messages are aggregated as follows:
(1) The $\oplus$ operator that combines pivots within a join group
is the weighted median.
The multiplicity function is given by the count of subtree answers
and the order by the ranking function.
The counts are also computed using message passing (see \Cref{sec:message_passing}).
(2) The $\otimes$ operator that combines pivots from different children
is simply the union of (partial) assignments to variables.

\begin{example}
Consider the binary-join $R_1(x_1,x_2), R_2(x_2,x_3)$ under full \ranksum.
Assume $R_1$ is the parent in the join tree with tuple weights $x_1 + x_2$,
while $R_2$ is the child with 
tuple weights $x_3$.
First, \PIVOT groups the $R_2$ tuples by $x_2$ and, for each group, 
it finds the median $x_3$ value. 
The message from $R_2$ to $R_1$ is a mapping from $x_2$ values to 
(1) the count of $R_2$ tuples that contain the $x_2$ value and 
(2) the median $x_3$ value over these tuples. 
Then, every tuple $r_1 \in R_1$ unions its $x_1, x_2$ values with the incoming $x_3$, 
obtaining a $\pivot(r_1) = (x1, x2, x3)$ tuple. 
To compute the final pivot, \PIVOT takes the median of these $\pivot(r_1)$ tuples, 
ranked by $x1+x2+x3$, and weighted by the count of $R_2$ tuples that contain the $x_2$ value.
\end{example}

\begin{example}
\Cref{fig:pivot} shows how an $R$-tuple computes its pivot under full \ranksum
for the example of \Cref{fig:count}.
Green values in brackets [.] represent the counts, 
while the orange assignments are the computed pivots for each tuple or join group.
From a leaf node like $S$ or $U$,
messages are simply the relation tuples, each with multiplicity 1.
To see how a pivot is computed within a join group, consider the $T$-node group.
The pivot of tuple $(1,6)$ is smaller than the pivot of tuple $(1,7)$
according to the ranking function
because $1 + 6+ 8 < 1 + 7 + 9$.
The weighted median selects the latter
because it has multiplicity 2 (that is the group size
for $x_4=6$ in the child $U$).
\end{example}

\begin{figure}[t]
\centering
\includegraphics[scale=0.5]{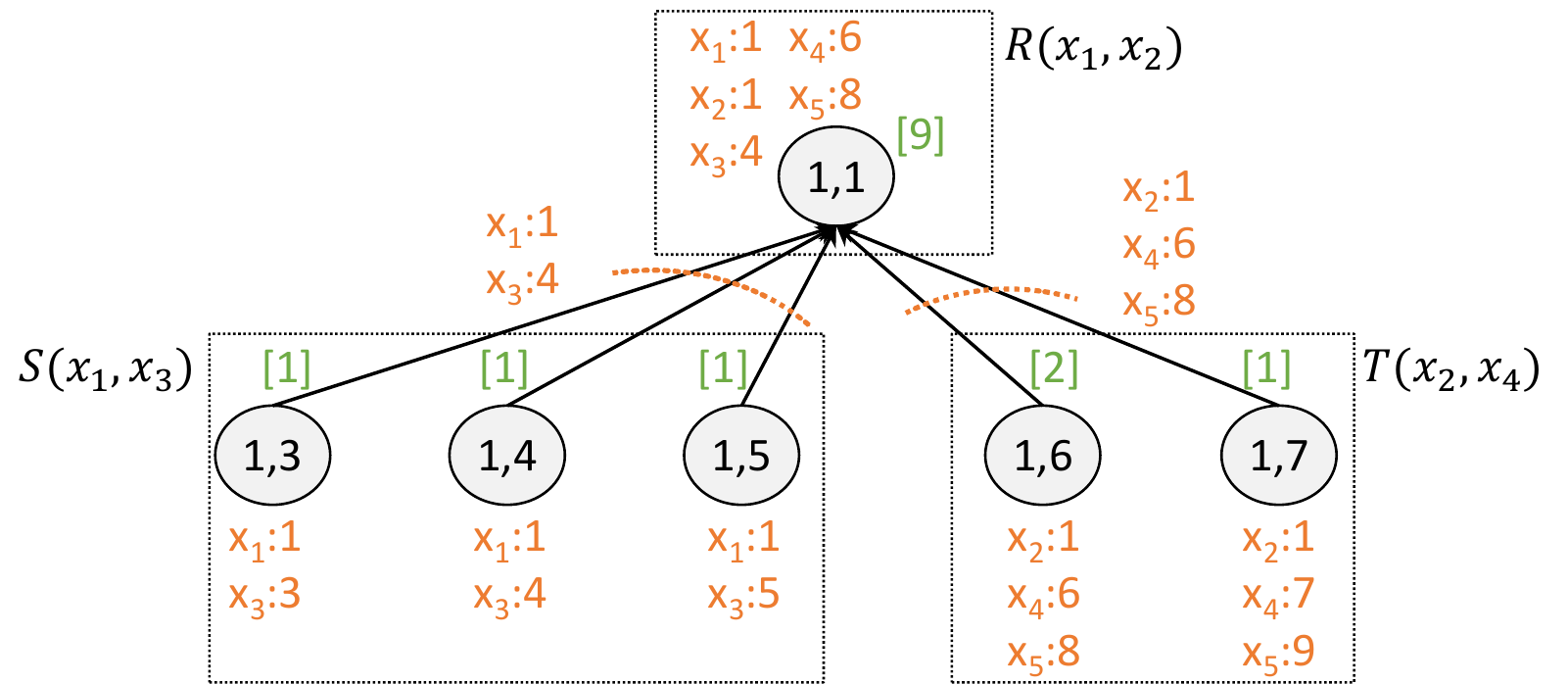}
\caption{Message passing for computing pivots
on the example \JQ and database of \Cref{fig:count}
under \ranksum with weights equal to attribute
values.}
\label{fig:pivot}
\end{figure}

\introparagraph{Pivot accuracy}
As we will show, applying our two operations (weighted median, union)
results in a loss of accuracy for the pivots, captured by the $c$ parameter.
The pivots computed for the leaf relations are the true medians (thus $\frac{1}{2}$-pivots),
but the $c$ parameter decreases as we go up the join tree.
Fortunately, this can be bounded by a function of the query size, making our final result a $c$-pivot 
with a $c$ value that is independent of the data size $n$.  
The algorithm keeps track of the $c$ value for every node
and upon termination, the $c$ value of the root is returned.

\introparagraph{Running time}
The time is linear in the database size.
The weighted median and count operations 
are only performed once for every join group
and both take linear time.
Each tuple is visited only once, and all operations per tuple
(e.g., number of child relations, finding the joining group, union)
depend only on the query size.

\subsection{Correctness}
\label{sec:pivot_correctness}

First, we show that \PIVOT returns a valid query answer.
The concern is that a variable $x$ may be assigned to
different values in the pivots that are unioned from different branches of the tree.
As we show next, this cannot happen because of the join tree properties.

\begin{lemma}
\label{lem:pivot_union_correctness}
Let $V$ be a join-tree node and $R_V$ the corresponding relation.
For all $t \in R_V$, the variable assignment $\pivot(t)$ computed by \PIVOT is a partial query answer for the subtree rooted at $V$.
\end{lemma}

Next, we show how the accuracy of the pivot is affected by repeated
weighted median and union operations.

\begin{lemma}\label{lem:median_of_pivots}
Given $r$ disjoint sets $Z_1, \ldots, Z_r$ equipped with a total order $\preceq$
and corresponding $c$-pivots $p_1, \ldots, p_r$,
then $\wm(\{p_1, \ldots, p_r\}, \beta)$ with $\beta(p_i) = |Z_i|$, for all $i \in [r]$ is a 
$\frac{c}{2}$-pivot for $Z_1 \cup\cdots\cup Z_r$.
\end{lemma}

\begin{lemma}
\label{lem:union_of_pivots}
Assume a join-tree node $V$, its corresponding relation $R_V$,
its children $V_1, \ldots, V_r$,
a \rankf ranking function,
and $c$-pivots $p_i, i \in r$
for the partial answers which are
rooted at 
$V_i$ and restricted to those that agree with $t$.
Then, $t \cup p_1 \cup \ldots \cup p_r$ is a $c^r$-pivot for the partial answers rooted at $t$.
\end{lemma}

With the three above lemmas, 
we can complete the proof of \Cref{lem:pivot} by induction on the join tree.

\section{Exact Trimmings}
\label{sec:exact}

We now look into trimmings for different types of inequality predicates
that arise in the partitioning step of our quantile algorithm (i.e., the \TRIM subroutine).
Our construction essentially removes these predicates from the query, while
ensuring that the modified query can only produce answers that satisfy them.

\subsection{\rankmin/\rankmax}

When the ranking function is \rankmin or \rankmax, 
then we need to trim predicates of the type $\min\{\Wvars\} < \wval$, $\wval \in \R$.

\begin{figure}[t]
\centering
\includegraphics[scale=0.4]{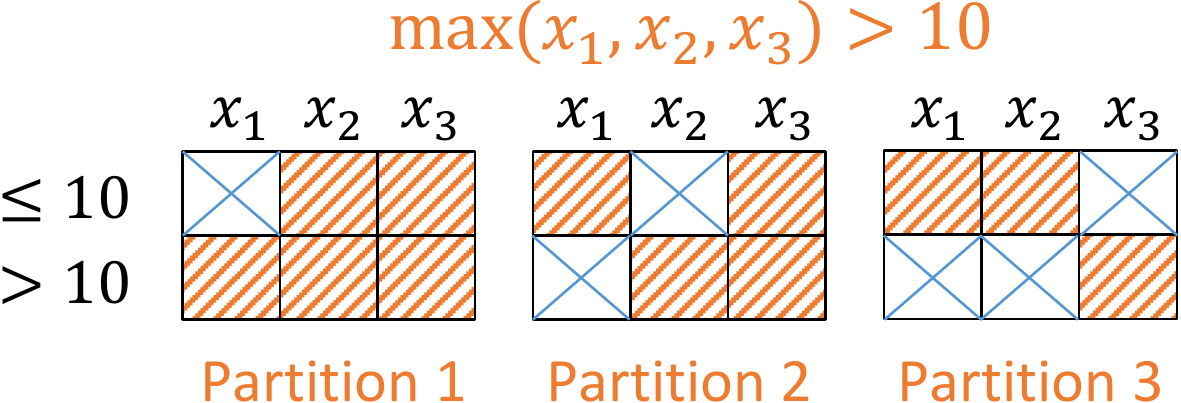}
\caption{Trimming $\max\{x_1, x_2, x_3\} > 10$ by expressing it 
as three disjoint partitions of the space of $x_1, x_2, x_3$ combinations, each one described by unary predicates (e.g., $x_1 > 10$).}
\label{fig:minmax}
\end{figure}

\begin{example}
Suppose the ranking function is \rankmax, 
the weighted variables are
$\Wvars = \{ x_1, x_2, x_3 \}$,
attribute weights are equal to database values,
and our pivot has weight $10$. 
To create the 
appropriate
partitions,
we trim predicates
$\max\{x_1, x_2, x_3\} < 10$
and $\max\{x_1, x_2, x_3\} > 10$.
Enforcing $\max\{x_1, x_2, x_3\} < 10$ is straightforward by removing
from the database all tuples with values greater than or equal to $10$ for either of the three variables.
For $\max\{x_1, x_2, x_3\} > 10$, there are three ways to satisfy
the predicate:
(1) $x_1 > 10$,
(2) $x_1 \leq 10 \land x_2 > 10$, or
(3) $x_1 \leq 10 \land x_2 \leq 10 \land x_3 > 10$.
These three cases are disjoint and cover all possibilities.
For each case, we create a fresh copy of the database and then enforce the predicates in linear time
by filtering the tuples.
Our \JQ over one of these three databases produces a partition of the answers that
satisfy the original inequality.
To return a single database and \JQ, we union the corresponding relations
and distinguish between the three partitions by appending a partition identifier to every relation.
\end{example}

Generalizing our example in a straightforward way, we show
that trimmings of such inequalities
exist for all acyclic \JQs.

\begin{lemma}
\label{lem:minmax_trim}
Given an acyclic \JQ $Q$,
variables $\Wvars \subseteq \var(Q)$,
weight functions $w_x: \dom \rightarrow \R$ for $x \in \Wvars$,
and $\wval \in \R$,
a trimming of
$\min_{x \in \Wvars} \w_x(x) < \wval$,
$\min_{x \in \Wvars} \w_x(x) > \wval$,
$\max_{x \in \Wvars} \w_x(x) < \wval$, or
$\max_{x \in \Wvars} \w_x(x) > \wval$
takes time $\bigO(n)$
and returns an acyclic \JQ.
\end{lemma}

Combining \Cref{lem:minmax_trim} together with \Cref{lem:pivot} and \Cref{lem:quantile_exact} gives us the following result:

\begin{theorem}
\label{th:minmax}
Given an acyclic \JQ over a database $D$ of size $n$, \rankmin or \rankmax ranking function,
and $\quant \in [0, 1]$,
the \%\JQ can be answered in  
time
$\bigO(n \log n)$.
\end{theorem}

\subsection{\ranklex}

For lexicographic orders, we provided~\cite{carmeli23direct}
an $\bigO(n)$ selection algorithm that can also be used for \%\JQs.
Our divide-and-conquer approach can recover this result up to a logarithmic factor, i.e., our \%\JQ algorithm runs in time $\bigO(n \log n)$.
To achieve that, we trim lexicographic inequalities, similarly to the case of \rankmin and \rankmax.

\begin{lemma}
\label{lem:lex_trim}
Given an acyclic \JQ $Q$,
variables $\Wvars = \{x_1, \ldots, x_r\} \subseteq \var(Q)$,
weight functions $w_x': \dom \rightarrow \R$ for $x \in \Wvars$,
and $\wval \in \R^r$,
a trimming of
$(\w_{x_1}'(x_1), \ldots, \w_{x_r}'(x_r)) <_\textrm{\ranklex} \wval$ or
$(\w_{x_1}'(x_1), \ldots, \w_{x_r}'(x_r)) >_\textrm{\ranklex} \wval$
takes time $\bigO(n)$
and returns an acyclic \JQ.
\end{lemma}

\subsection{Partial \ranksum}

We now consider the case of \ranksum.
While we previously gave a dichotomy~\cite{carmeli23direct} for all self-join-free
\JQs,
this result is limited to full \ranksum.
We now provide a more fine-grained dichotomy 
where certain variables may not participate in the ranking.
For example, the 3-path \JQ $R_1(x_1, x_2), R_2(x_2, x_3), R_3(x_3, x_4)$ would be classified as intractable by the prior dichotomy,
yet with weighted variables $\Wvars = \{x_1, x_2, x_3 \}$,
we show that it is in fact tractable.

The positive side of our dichotomy requires
a trimming of additive inequalities.
We rely on a known algorithm that can be applied whenever
the \ranksum variables appear in adjacent join-tree nodes.

\begin{lemma}[\cite{tziavelis21inequalities}]
\label{lem:partial_sum_trim}
Given a set of variables $\Wvars$,
let $\mathcal{Q}$ be the class of acyclic \JQs $Q$ for which there
exists a join tree where $\Wvars \subseteq \var(Q)$ belong to adjacent join-tree nodes.
Then, for all $Q \in \mathcal{Q}$,
weight functions $w_x: \dom \rightarrow \R$ for $x \in \Wvars$,
and $\wval \in \R$, 
a trimming of
$\sum_{x \in \Wvars} \w_x(x) < \wval$ or
$\sum_{x \in \Wvars} \w_x(x) > \wval$
takes time $\bigO(n \log n)$
and returns a \JQ $Q' \in \mathcal{Q}$.
\end{lemma}

We are now in a position to state our dichotomy:

\begin{theorem}
\label{th:partial_sum}
Let $Q$ be a self-join-free \JQ, $\calH(Q)$ its hypergraph, and
$\Wvars \subseteq \var(Q)$ the variables of a \ranksum{} ranking function.
\begin{itemize}

\item
If $\calH(Q)$ is acyclic, any set of independent variables of $\Wvars$ is of size at most $2$, and any chordless path between two $\Wvars$ variables is of length at most $3$, then \%\JQ can be answered in $\bigO(n \log^2 n)$.

\item Otherwise, \%\JQ cannot be answered in $\bigO(n \polylog n)$,
assuming \ThreeSUM{} and \hyperclique{}.
\end{itemize}

\end{theorem}

We note that the positive side also applies to \JQs with self-joins.

\section{Approximate Trimming for \ranksum}
\label{sec:approx}

We now move on to devise
an $\epsilon$-lossy trimming for additive inequalities
in order to get a deterministic approximation
for \ranksum and \JQs  
beyond those covered by \Cref{th:partial_sum}.

\begin{lemma}
\label{lem:approx_trim}
Given an acyclic \JQ $Q$, 
variables $\Wvars \subseteq \var(Q)$,
weight functions $w_x: \dom \rightarrow \R$ for $x \in \Wvars$,
$\wval \in \R$,
and $\epsilon \in (0, 1)$,
an $\epsilon$-lossy trimming of
$\sum_{x \in \Wvars} \w_x(x) < \wval$ or
$\sum_{x \in \Wvars} \w_x(x) > \wval$
takes time $\bigO(\frac{1}{\epsilon^2} n \log^2 n \log \frac{n}{\epsilon})$
and returns an acyclic \JQ.
\end{lemma}

This, together with \Cref{lem:quantile_approx,lem:pivot}
gives us the following:

\begin{theorem}
Given an acyclic \JQ $Q$ over a database $D$ of size $n$, \ranksum ranking function,
$\quant \in [0, 1]$,
and $\epsilon \in (0, 1)$,
the $\epsilon$-approximate \%\JQ can be answered in time
$\bigO(\frac{1}{\epsilon^2} n \log^5 n \log \frac{n}{\epsilon})$.
\end{theorem}

To achieve the trimming, we adapt an algorithm of Abo-Khamis et al.~\cite{khamis21approx},
which we refer to as \APPROXCOUNT.
It computes an approximate count (or more generally, a semiring aggregate) over the answers
to acyclic \JQs with additive inequalities. 
In contrast, we need not only the count of answers,
but an efficient \emph{relational representation} of them as \JQ answers over a new database.
We only discuss the case of less-than ($<$), since the case of greater-than ($>$) is symmetric.
The detailed pseudocode is given in \Cref{sec:approx_details}.

\introparagraph{Message passing}
\APPROXCOUNT uses message passing (see \Cref{sec:message_passing}).
We first describe the exact, but costly, version of the algorithm.
The message sent by a tuple is a multiset containing the (partial) sums
of partial query answers in its subtree.
Messages are aggregated as follows:
(1) The $\oplus$ operator that combines multisets within a join group
is multiset union ($\uplus$).
(2) The $\otimes$ operator that combines multisets from different children
is pairwise summation (applied as a binary operator).
The messages emitted by the root-node tuples contain all query-answer sums, 
which can be leveraged to count the number of answers that satisfy the inequality.

\introparagraph{Sketching}
Sending all possible sums up the join tree is intractable since, in the worst case, 
their number is equal to the number of \JQ answers.
For this reason, \APPROXCOUNT applies \emph{sketching} to compress the messages.
The basic idea is to replace different elements in a multiset with the same element; 
the efficiency gain is due to the fact that an element $s$ that appears $r$ times can be represented as $s \times r$.
In more detail, the multiset elements are split into buckets, and subsequently, each element within a bucket is replaced by the maximum element of the bucket.
A sketched multiset $L$ is denoted by $\mathbb{S}_\epsilon(L)$,
where $\epsilon$ is a parameter that determines the number of buckets.
Let 
$\countlessthan{L}{\wval}$
be the number of elements of $L$ that are less than $\wval \in \R$.
By choosing buckets appropriately, we can
guarantee that $\countlessthan{\mathbb{S}_\epsilon(L)}{\wval}$
is close to $\countlessthan{L}{\wval}$ for all possible $\wval$.
 
\begin{lemma}[$\epsilon$-Sketch~\cite{khamis21approx}]
\label{lem:sketch}
For a multiset $L \in \N^\R$ and $\epsilon \in (0,1)$, we can construct a sketch $\mathbb{S}_\epsilon(L)$ 
with $\bigO(\log_{1 + \epsilon} |L|)$ distinct elements
such that for all $\wval \in \R$, we have
$(1 - \epsilon) \countlessthan{L}{\wval} \leq \countlessthan{\mathbb{S}_\epsilon(L)}{\wval} \leq \countlessthan{L}{\wval}$.
\end{lemma}

\APPROXCOUNT sketches all messages and bounds the error incurred by the two message-passing operations ($\oplus, \otimes$).

\introparagraph{Relational representation}
Our goal is to construct a relational representation of the \JQ answers
which satisfy the inequality that we want to trim.
The key idea is to \emph{embed the sums} contained in the messages of \APPROXCOUNT
into the database relations 
so that each tuple stores a \emph{unique sum} and all answers in its subtree approximately have that sum.
The reasoning behind this is that we can then remove from the database the root tuples whose associated sum does not satisfy the inequality.
However, in \APPROXCOUNT a message is a multiset of sums, and hence 
the main technical challenge we address below is how to achieve a unique sum per tuple (and its subtree).

\introparagraph{Separating sums}
Let $\sumset(t')$ be the message sent by a tuple $t'$ in a child relation $S$.
Then, according to \APPROXCOUNT, a tuple $t$ in the parent relation $R$
receives a message $\sumset(b) = \mathbb{S}_{\epsilon'}(\uplus_{t' \in b} \sigma(t'))$
for some $\epsilon'$ and join group $b$.
We separate the sums in this multiset by creating  
a number of copies of $t$,
equal to the number of distinct sums
in $\sumset(b)$.
Each $t$-copy is associated with a unique bucket $e$, described by a
sum value $e_s$ and a multiplicity $e_m$.
To avoid duplicating query answers,
we restrict each $t$-copy to join only with the \emph{source tuples} of its associated bucket $e$,
i.e., the child tuples $t' \in S$ whose messages were assigned to bucket $e$ during sketching.

\begin{figure}[t]
\centering
\includegraphics[scale=0.37]{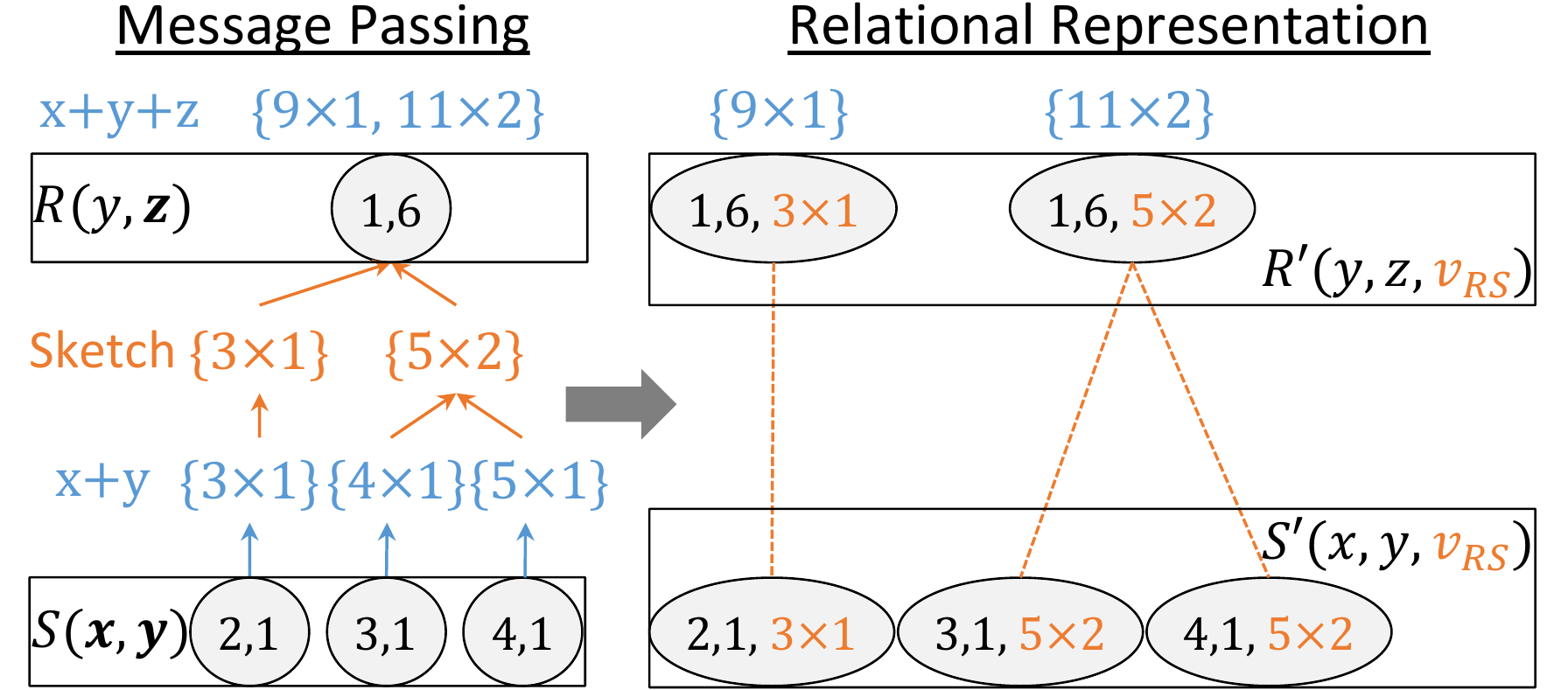}
\caption{Example of how we use the message passing framework~\cite{khamis21approx} to create a relational representation of the query answers that satisfy an inequality $x+y+z < \lambda$.}
\label{fig:sum}
\end{figure}

\begin{example}
\label{ex:sum}
\Cref{fig:sum} illustrates how we embed messages into the database relations
for a leaf relation $S$ and a parent relation $R$ with no other children,
and assuming weights equal to attribute values.
The messages from $S(x,y)$ to $R(y, z)$ are the sums $x+y$ (because $y$-weights are assigned to $S$).
After sketching their union using two buckets, sums $4$ and $5$ are both mapped to $5$;
we keep track of this with a multiplicity counter (shown as $\times 2$), reflecting the number of answers in the subtree.
Upon reaching relation $R$, the weight of the $R$-tuple (which is the $z$-value)
is added to all sums.
For our relational representation,
we duplicate the $R$-tuple and associate each copy with a unique sum.
A copy corresponds to a bucket in the sketch,
so we can trace its ``source'' $S$-tuples,
i.e., those that belong to that bucket.
Instead of joining with all $S$-tuples like before,
a copy now only joins with the source $S$-tuples of the bucket
via a new variable $v_{RS}$ that stores the sum and the multiplicity.
\end{example}

\introparagraph{Adjusting the sketch buckets}
An issue we run into with the approach above is that the sum sent by a single tuple
may be \emph{assigned to more than one bucket} during sketching.
To see why this is problematic, consider a tuple $t'$ that sends $\sigma(t') = 5 \times 10$.
For simplicity, assume these are the only values to be sketched and that the two buckets 
contain $5 \times 3$ and $5 \times 7$.
With these buckets, we will create two copies of a tuple $t$ in the parent and both will join with $t'$,
because $t'$ is the source tuple for both buckets.
By doing so, we have effectively doubled the number of (partial) query answers
since there are now $10$ answers in the subtree of \emph{each} copy.
To resolve this issue, we need to guarantee that all elements in $\sigma(t')$
are assigned to the same bucket.
We adjust the sketching $\mathbb{S}_\epsilon(L)$ of a multiset $L$ as follows.
The $r$ buckets are determined by an increasing sequence of $r+1$ 
indexes on
an array that contains $L$ sorted.
The first index is $0$ and the last index is equal to $|L|-1$
(where $|L|$ takes into account the multiplicities).
Consider three consecutive indexes $i, j, k$
which define two consecutive buckets where the 
values at the borders are the same, i.e., 
$L[j-1] = L[j]$.
Let $j'$ and $j''$ be the smallest and largest indexes that contain $L[j]$
in the buckets $i-j$ and $j-k$, respectively.
We replace the indexes $i, j, k$ with $i, j', j''+1, k$
(and if 2 consecutive indexes are the same, then we remove that bucket).
As a result,
all values $L[j]$ from these two buckets now fall into the same bucket.
We repeat this process for every two consecutive buckets.
This can at most double the number of buckets, which,
as we show, does not affect the guarantee of \Cref{lem:sketch}.

\introparagraph{Binary join tree}
The running time of our algorithm (in particular the logarithmic-factor exponent)
depends on the maximum number of children of a join-tree node.
This is because we handle each parent-child node pair separately,
and each child results in the parent relation growing 
by the size of the messages, which is logarithmic.
To achieve the time bound stated in \Cref{lem:approx_trim},
we impose a binary join tree, i.e., every node has at most two children.
Such a join tree can be constructed by creating copies of a node that has
multiple children,
connecting these copies in a chain,
and distributing the original children among them.
In the worst case, this doubles the number of nodes in the join tree
(hence the number of relations that we materialize), but it does not affect the data complexity.

\section{Conclusions}
\label{sec:conclusions}

We can often answer quantile queries over joins of multiple relations much more efficiently than it takes to materialize the result of the join. Here, we adopted \e{quasilinear time} as our yardstick of efficiency. With our divide-and-conquer technique,
we recovered known results (for lexicographic orders) and established new ones (for partial sums, minimum, and maximum). We also showed how the approach can be adapted for deterministic approximations. 

We restricted the discussion to \JQs, that is, full Conjunctive Queries (CQs), and left open the treatment of non-full CQs (i.e., joins with projection). 
Most of our algorithms apply to CQs that are acyclic and free-connex, 
but it is not yet clear to us whether our results cover all tractable cases (under complexity assumptions).
More general open directions are the generalization of the challenge to \e{unions} of CQs, and the establishment of nontrivial algorithms for general CQs beyond the acyclic ones.

\begin{acks}
This work was supported in part by NSF
under award numbers IIS-1762268 and IIS-1956096. 
Benny Kimelfeld was supported by the German Research Foundation (DFG) Project 412400621.
Nikolaos Tziavelis was supported by a Google PhD fellowship.
\end{acks}

\bibliographystyle{ACM-Reference-Format}
\bibliography{bibliography.bib}

\appendix

\section{Nomenclature}

\begin{table}[h]
\centering
\small
\begin{tabularx}{\linewidth}{@{\hspace{0pt}} >{$}l<{$} @{\hspace{2mm}}X@{}} %
\hline
\textrm{Symbol}		& Definition 	\\
\hline
    Z               & generic set \\
    L               & generic multiset \\
	R,S,T,R_1,R_2		& relation \\
	V, V_1, V_2 & atom/hyperedge/node of join tree \\
	\calS       & schema \\
	D           & database (instance) \\
        n           & size of $D$ (number of tuples) \\	
	\dom        & database domain \\
	t			& tuple \\	
	x,y,z,u,v		& variable \\	
	Q			& Join Query (\JQ)	\\
    \ell		& number of atoms in a \JQ \\
    \var(Q)     & variables contained in $Q$ \\
    Q(D)        & set of answers of $Q$ over $D$ \\
    (Q \wedge P)(D) & subset of $Q(D)$ answers that satisfy a predicate $P$ \\
    q \in Q(D)  & query answer \\
	\calH(Q) = (V, E)	& hypergraph associated with query $Q$ \\
	T   & join tree \\

    \quant    & fraction in $[0, 1]$ used to ask for a quantile \\
 
    \w    & weight function for query answers \\
    \Wdom   & domain of weights \\
    \Wvars  & subset of variables that participate in the ranking \\
    \w_x     & input weight function for variable $x$: $\dom \rightarrow \Wdom$\\
    \w_R     & input weight function for tuples of relation $R$: $\dom^{\ar(R)} \rightarrow \Wdom$ where $\ar(R)$ is the arity of $R$\\
    \Waggr    & aggregate function that combines input weights to derive weights for query answers \\
    \w(\Wvars) & aggregate function applied on the weighted variables, i.e., $\Waggr(\{\w_x(x) | x \in \Wvars \})$ \\
    \lambda & a weight from $\Wdom$ \\

\hline
\end{tabularx}
\end{table}

\section{Details of Divide-and-Conquer Framework}
\label{sec:details_quantile}

\Cref{alg:quantile_approx} returns the desired (approximate) quantile for a given \JQ,
database, and ranking function, as
presented in \Cref{sec:quantile-alg}.
The exact version is obtained by simply setting $\epsilon=0$.

\subsection{Proof of \Cref{lem:quantile_approx}}

Let $Q_i$ and $D_i$ be the \JQ and database at the start of iteration $i$, with $i \geq 0$
(these are the variables $Q'$ and $D'$ in \Cref{alg:quantile_approx}).
Even though trimmings may increase the size of our queries by a constant factor (by the definition of trimming),
all queries $Q_i$ have constant size.
This is because we start every iteration with the original query $Q$
and every other query we construct is the result of applying at
most two consecutive trimmings.

First, we bound the number of iterations.
Iteration $i$ is guaranteed to eliminate at least $c|Q_i(D_i)|$ query answers
because (i) we select a $c$-pivot to partition
and (ii) the lossy trimmings may result in more query answers being eliminated than they should, but never less.
Consequently, at the beginning of iteration $i$, we have at most $(1-c)^i|Q(D)|$ query answers remaining.
The number of query answers is bounded by $n^\ell$ where $\ell$ is the number of atoms in $Q$.
If $I$ is the total number of iterations, then
$I \leq \lceil \log_{1/(1-c)} |Q(D)| \rceil \leq \lceil \ell  \log_{1/(1-c)} n \rceil = \bigO(\log n)$ since $c$ and $\ell$ are constants.

Second, we show that the returned answer is an $\epsilon$-approximate quantile.
The less-than partition $Q_{\texttt{lt}}(D_{\texttt{lt}})$ is constructed by trimming the inequalities $\w(\Wvars) \prec \w(p)$
and 
$\w(\Wvars) \succ \texttt{low}$ with some error $\epsilon'$, where $\texttt{low}$ lower-bounds the weights of the candidate query answers.
Because these two trimmings are lossy, we ``lose'' a number of query answers which are at most 
$2 \epsilon' |Q(D)|$.
These are the answers that satisfy the inequalities, but do not appear in $Q_{\texttt{lt}}(D_{\texttt{lt}})$.
As \Cref{fig:approximate_trimming} illustrates, all answers not contained in the less-than or greater-than partition,
including these lost query answers, are assumed to be contained in the equal-to partition which we do not explicitly count.
We now bound the distance between the desired index and the index of the answer that our algorithm returns.
Each iteration $i$ starts with an index $k_i$
and results in a new index $k_i'$,
which the following iteration is asked to retrieve (or, in case this is the last iteration, the index that is returned).
Note that in \Cref{alg:quantile_approx} the variable $k$ indexes the subarray of query answers that are currently candidates;
thus it is offset by the index of the answer with weight $\textrm{low}$.
Here, the indexes $k_i$ and $k_i'$ refer to the original array that contains all the query answers.
Suppose that $k_i$ falls into the less-than partition at the beginning of the iteration.
Then, if $k_i'$ is different than $k_i$, it has to be a higher index 
because of lost query answers that precede it and which are moved to the middle equal-to partition
(see \Cref{fig:approximate_trimming}).
Thus, $|k_i - k_i'| \leq 2 \epsilon' |Q(D)|$.
If $k_i$ falls into the equal-to partition, then we still choose that partition and return the pivot
because the size of the partition can only increase from the lossy trimmings.
For the greater-than partition, the analysis is symmetric to lower-than
since the lossy trimmings of the latter do not affect the indexes of the former.
To conclude, the accumulated absolute error is
$I \cdot 2 \epsilon' |Q(D)| \leq 2 \lceil \ell \log_{1/(1-c)} n \rceil \epsilon' |Q(D)|$.
To obtain an $\epsilon$-approximate quantile of $Q(D)$,
we set $\epsilon' = \frac{\epsilon}{2 \lceil \ell \log_{1/(1-c)} n \rceil}$.

Finally, we prove the running time.
Since our trimmings return acyclic \JQs,
the answers of all queries we construct
can be counted in linear time. 
Thus, the running time per iteration is $\bigO(g_p(n) + 4 g_t(n, \epsilon') + n)$
which is $\bigO(\max \{ g_p(n), g_t(n, \epsilon') \}$ since $g_p(n)$ and $g_t(n, \epsilon')$
are necessarily $\Omega(n)$.

We note that this proof also covers \Cref{lem:quantile_exact} since
\Cref{lem:quantile_approx} is a stronger version of it.

\begin{algorithm}[t]
\footnotesize
\textbf{Input}: acyclic \JQ $Q$, database $D$, ranking function $(\w, \preceq)$, quantile $\quant$,
approximation bound $\epsilon$\\
\textbf{Output}: the $\quant$-th quantile of $Q(D)$\\

\algocomment{Calculate desired index}\;
Determine $|Q(D)|$ and set $k = \lfloor \quant \cdot |Q(D)| \rfloor$ (with zero-indexing)\;

\algocomment{Calculate parameter for trimming ($\epsilon=\epsilon'=0$ for exact)}\;
$\epsilon' = \frac{\epsilon}{2 \lceil \ell \log_{1/(1-c)} n \rceil}$

\algocomment{Each iteration modifies $Q'(D')$ by bringing $\texttt{low}$ and $\texttt{high}$ closer}\;
$(Q', D', \texttt{low}, \texttt{high}) = (Q, D, \bot, \top)$\;

\While{$|Q'(D')| > |D|$}
{
    \algocomment{Select a $c$-pivot $p$}\;
    $(p, c) = \PIVOT(Q', D', (\w, \preceq))$\;

    \algocomment{Partition}\;
    $(Q_{\texttt{lt}}, D_{\texttt{lt}}) = \TRIM(Q, D, \w(\Wvars) \prec \w(p), \epsilon')$\;
    $(Q_{\texttt{lt}}, D_{\texttt{lt}}) = \TRIM(Q_{\texttt{lt}}, D_{\texttt{lt}}, \w(\Wvars) \succ \texttt{low}, \epsilon')$\;

    $(Q_{\texttt{gt}}, D_{\texttt{gt}}) = \TRIM(Q, D, \w(\Wvars) \succ \w(p), \epsilon')$\;
    $(Q_{\texttt{gt}}, D_{\texttt{gt}}) = \TRIM(Q_{\texttt{gt}}, D_{\texttt{gt}}, \w(\Wvars) \prec \texttt{high}, \epsilon')$\;

    \algocomment{Choose partition}\;
    Set $|Q_{\texttt{eq}}(D_{\texttt{eq}})|$ to 
    $|Q'(D')| - |Q_{\texttt{lt}}(D_{\texttt{lt}})| - |Q_{\texttt{gt}}(D_{\texttt{gt}})|$ \;
    
    \If{$k < |Q_{\texttt{lt}}(D_{\texttt{lt}})|$}
    {
        $(Q', D', \texttt{high}) = (Q_{\texttt{lt}}, D_{\texttt{lt}}, \w(p))$\;
    }
    \ElseIf{$k < |Q_{\texttt{lt}}(D_{\texttt{lt}})| + |Q_{\texttt{eq}}(D_{\texttt{eq}})|$}
    {
        \Return $p$\;
    }
    \Else
    {
        $(Q', D', \texttt{lt}) = (Q_{\texttt{gt}}, D_{\texttt{gt}}, \w(p))$\;
        $k = k - |Q_{\texttt{lt}}(D_{\texttt{lt}})| - |Q_{\texttt{eq}}(D_{\texttt{eq}})|$\;
    }
}

Materialize and sort $Q'(D')$\;
\Return answer at index $k$ in $Q'(D')$\;

\caption{Pivoting Algorithm}
\label{alg:quantile_approx}
\end{algorithm}

\begin{figure}[t]
\centering
\includegraphics[scale=0.75]{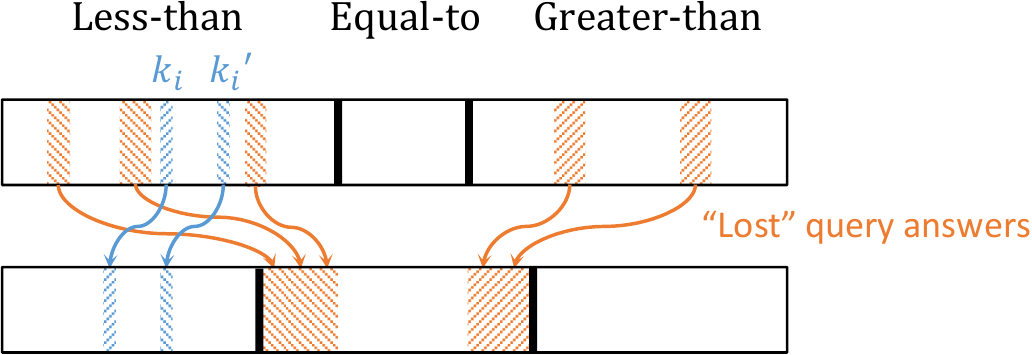}
\caption{Proof of \Cref{lem:quantile_approx}: Query answers that are ``lost'' due to lossy trimmings are implicitly moved to the equal-to partition (middle).
Consequently, index $k_i$ in the less-than partition contains an element that was previously at a higher index $k_i'$, but $k_i'-k_i$ is bounded by the number of lost answers.}
\label{fig:approximate_trimming}
\end{figure}

\section{Details of Choosing a Pivot}
\label{sec:pivot_details}

\Cref{alg:pivot} shows the algorithm that returns a $c$-pivot for a given \JQ,
database, and ranking function, as
presented in \Cref{sec:pivot}.

\subsection{Proof of \Cref{lem:pivot_union_correctness}}

Let $b_1, \ldots, b_r$ be the joining groups from the children $V_1, \ldots, V_r$ of $V$,
We show by induction on the join tree that $\pivot(b_1), \ldots, \pivot(b_r)$, and $t$ all agree on their common variables.
For the leaf relations, we have no children and $\pivot(t)$ is initialized to $t$.
For the inductive step, let $x$ be a common variable between two children $V_i$ and $V_j$ of $V$,
$i,j \in [r]$.
Because of the running intersection property of the join tree, 
$x$ also needs to appear in the parent $V$.
Since the groups $b_i, b_j$ join with $t$,
all their tuples necessarily assign value $t[x]$ to variable $x$.
We show that $\pivot(b_i)$ also assigns $t[x]$ to $x$ and the case of $\pivot(b_j)$ is similar.
We have that $\pivot(b_i) = \pivot(t_i)$ for some tuple $t_i \in b_i$
where $\pivot(t_i)$ is picked as the weighted median of the group.
From the inductive hypothesis, $\pivot(t_i)$ needs to agree with $t_i$ on the value of $x$
which we argued is equal to $t[x]$.

\subsection{Proof of \Cref{lem:median_of_pivots}}

Without loss of generality, let the indexing of the $r$ sets be consistent with
the ordering of their $c$-pivots, i.e., $p_i \preceq p_j$
for $Z_i, Z_j, 1 \leq i \leq j \leq r$.
Let $p_m$ be the weighted median, selected from set $Z_m$ for some $m \in [r]$.
We prove that $p_m$ is greater than or equal to (according to $\preceq$) 
at least $\frac{c}{2} |Z_1 \cup Z_2 \cup \ldots \cup Z_r|$ elements,
and the case of less than or equal to is symmetric.
Because of the indexing we enforced, we know that $p_i \preceq p_m$ for all $i \in [m]$.
Combining that with the definition of a $c$-pivot (for $p_i$),
we obtain that
$p_m$ is greater than or equal to at least $c|Z_i|$ elements of $|Z_i|$,
or $c \sum_{i \in [1, m]} |Z_i|$ in total. 
Now, because the median is weighted by the set sizes and there is no overlap between their elements,
$\sum_{i \in [1, m]} |Z_i| \geq |Z_1 \cup Z_2 \cup \ldots \cup Z_r| / 2$.
Thus, $p_m$ is greater than or equal to at least 
$c|Z_1 \cup Z_2 \cup \ldots \cup Z_r| / 2$ elements
of $Z_1 \cup Z_2 \cup \ldots \cup Z_r$.

\subsection{Proof of \Cref{lem:union_of_pivots}}

Let $M$ be the partial answers rooted at $t$, and let $M_i$ be the partial answers
rooted at 
$V_i$ and restricted to those that agree with $t$ for all $i \in [r]$.
We only show that $t \cup p_1 \cup \ldots \cup p_r$ is greater than or equal to
at least 
$c^r|M|$ partial answers,
since the case of less than or equal to is symmetric.
For $i \in [r]$, let $L_i$ be the subset of $M_i$ answers
that are less than or equal to $p_i$.

We first show that 
$\w(t \cup q_1 \cup \ldots \cup q_r) \preceq 
\w(t \cup p_1 \cup \ldots \cup p_r)$
whenever $q_i \in L_i, i \in [r]$.
We know that $\w(q_i) \preceq \w(p_i)$ for all $i \in [r]$.
We proceed inductively in $i$, showing that 
$\w(q_1 \cup \ldots \cup q_i) \preceq \w(p_1 \cup \ldots \cup p_i)$.
The inductive hypothesis is that
$\w(q_1 \cup \ldots \cup q_{i-1}) \preceq \w(p_1 \cup \ldots \cup p_{i-1})$.
Each of these two terms is an aggregate over the values of variables
$U_1 \cup \ldots \cup U_{i-1}$ mapped to their weights,
where $U_i$ is the subset of weighted variables $\Wvars$ that appear in the subtree rooteed at node $V_i$.
For example, $\w(q_1 \cup \ldots \cup q_{i-1})$ is 
$\Waggr(\{\w_x((q_1 \cup \ldots \cup q_{i-1})[x]) \mid x \in U_1 \cup \ldots \cup U_{i-1} \})$.
By subset-monotonicity,
we can add to both aggregates the weighted values of $p_i$
without changing the inequality, i.e., we obtain
$\w(q_1 \cup \ldots \cup q_{i-1} \cup p_i) \preceq \w(p_1 \cup \ldots \cup p_{i-1} \cup p_i)$ (A).
With a similar argument of subset-monotonicity, we start from $\w(q_i) \preceq \w(p_i)$ and add to both sides the weighted values of $q_1 \cup \ldots \cup q_{i-1}$ to obtain
$\w(q_1 \cup \ldots \cup q_{i-1} \cup q_i) \preceq \w(q_1 \cup \ldots \cup q_{i-1} \cup p_i)$ (B).
(A) and (B) together prove the inductive step.
To complete the first part of the proof, 
we add to both aggregates the weighted values of $t$ 
(that do not appear in any child)
to obtain $\w(t \cup q_1 \cup \ldots \cup q_r) \preceq 
\w(t \cup p_1 \cup \ldots \cup p_r)$,
again by subset-monotonicity.

Next, notice that there are $|L_1 \times \ldots \times L_r|$ partial answers of the form $t \cup q_1 \cup \ldots \cup q_r$ with $q_i \in L_i, i \in [r]$.
Since every $L_i$ comprises of elements that are less than
or equal to a $c$-pivot, we have $|L_i| \geq c|M_i|$. Also notice that $|M|=\prod_{i \in [r]} |M_i|$. Overall, we have that $|L_1 \times \ldots \times L_r| \geq \prod_{i \in [r]} c|M_i|
= c^r |M|$, and so $t \cup p_1 \cup \ldots \cup p_r$ is greater than or equal to
at least 
$c^r|M|$ partial answers.

\begin{algorithm}[t]
\footnotesize
\textbf{Input}: acyclic \JQ $Q$, database $D$, ranking function $(\w, \preceq)$\\
\textbf{Output}: a $c$-pivot of $Q(D)$ ordered by $\preceq$, and the value of $c$\\

Convert attribute weights to tuple weights\;
Construct a join tree $T$ of $Q$ with artificial root $V_0 = \{ t_0 \}$\;
Materialize a relation for every $T$-node and group it by the variables it
has in common with its parent node\;
Initialize $\pivot(t) = t, \cnt(t) = 1$ for all tuples $t$ of all relations\; 
Initialize $c(R) = 1$ for leaf relations $R$\; 

\For{relation $R$ in bottom-up order of $T$}
{   
    \If{$R$ is not leaf}
    {
        $S_1, \ldots, S_r$ = children of $R$\;
        $c(R) = \frac{c(S_1)}{2} \times \ldots \times \frac{c(S_r)}{2}$\;
    }
    
    \For{tuple $t \in R$, child $S$ of $R$}
    {
            $b$ = join group of $S$ that agrees with the values of $t$\;
            \algocomment{Compute the weighted median (and the count of subtree answers) the first time we visit this group}\;
            \If{$\pivot(b)$ not already computed \label{alg:pivot-memoization}}
            {
                $\pivot(b)$ = $\wm_\preceq(\{ \pivot(t') \; | \; t' \in b\}, \beta)$ with $\beta(\pivot(t')) = \cnt(t')$ \label{alg:pivot-wmed}\;
                $\cnt(b) = \sum_{t' \in b} \cnt(t') $ \label{alg:pivot-cnt}\;
            }
            \algocomment{Combine results from different branches of the join tree}\;
            $\pivot(t) = \pivot(t) \cup \pivot(b)$ \label{alg:pivot-union}\;
            $\cnt(t) = \cnt(t) \times \cnt(b)$\;
    }
}
\Return $\pivot(t_0)$, $c(V_0)$\;

\caption{\PIVOT}
\label{alg:pivot}
\end{algorithm}

\section{Details of Exact Trimmings}
\label{sec:exact_trim_details}

\subsection{Proof of \Cref{lem:minmax_trim}}
\label{sec:proof_minmax}
We always start by creating fresh copies of relations to eliminate self-joins from $Q$.
This ensures that every column in the database corresponds to a unique variable,
avoiding situations like $R(x,y), R(y,x)$.

First, consider $\max_{x \in \Wvars} \w_x(x) < \wval$.
We scan the given database $D$ once and if a tuple $t$
contains a value $t[x]$ with $\w_x(t[x]) \geq \wval$ for a variable $x \in \Wvars$,
then we remove $t$ from the database.
This process removes precisely the answers $q \in Q(D)$ 
that do not satisfy the predicate,
since, for the maximum to be greater than or equal to $\wval$,
at least one variable needs to map to such a weight.
The \JQ we return is $Q$ itself.
The case of $\min_{x \in U} \w_x(x) > \wval$ is symmetric.

Second, consider $\max_{x \in \Wvars} \w_x(x) > \wval$.
\Cref{alg:max-trim} shows the pseudocode of \TRIM for this case.
If there are $r$ variables in $\Wvars$,
then we create $r$ databases, each enforcing condition $P_i$,
which is a conjunction of unary predicates.
The conditions $P_i$ partition the space of possible $\Wvars$ values
that satisfy $\max_{x \in \Wvars} \w_x(x) > \wval$.
To return a single database $D'$, 
we union together the copies of each relation 
and separate the different databases with a partition identifier $i \in [r]$.
This identifier is added as a variable $x_p$ to all atoms of the returned \JQ $Q'$.
As a consequence, each query answer of the returned $Q'$ can only draw values from database tuples that belong to the same partition.
The bijection from $Q'(D')$ to $Q(D)$ simply removes the variable $x_p$.
Since $r$ does not depend on $D$, the entire process can be done in linear time.
Furthermore, $Q'$ remains acyclic because every join tree of $Q$
is also a join tree of $Q'$ by adding $x_p$ to all nodes.
The case of $\min_{x \in \Wvars} \w_x(x) < \wval$ is symmetric.

\begin{algorithm}[t]
\footnotesize
\SetAlgoLined
\LinesNumbered
\textbf{Input}: acyclic \JQ $Q$, database $D$, predicate $\max(\Wvars) > \wval$\\
\textbf{Output}: acyclic \JQ $Q'$, database $D'$ \\

    $x_1, \ldots, x_r = \Wvars$\;
    $(Q', D') = (Q, \emptyset)$\;

    \algocomment{Construct the new \JQ}\;
    Eliminate self-joins from $Q'$ by materializing new relations in $D$\;
    Add the same variable $x_p$ to all the atoms and the head of $Q'$\;
    
    \algocomment{Create $r$ databases}\;
    \For {$i$ from 1 to $r$}
    {
        \algocomment{Each $P_i$ is a conjuntion of unary predicates}\;
        $P_i = \{ \w_{x_1}(x_1) \leq \lambda, \ldots, \w_{x_{i-1}}(x_{i-1}) \leq \lambda, 
                \w_{x_i}(x_i) > \lambda \}$\;

        $D_i =$ copy of $D$ with conditions $P_i$ applied\;

        \algocomment{An identifier separates the answers from different $D_i$ after the union}\;
        Add the column $x_p$ with value $i$ to all relations of $D_i$\;
    }

    \algocomment{Union the databases into one}\;
    \For {relation $R^D$ in $D$}
    {   
        Add to $D'$ relation $R^{D'} = \bigcup_{i \in [r]} R^{D_i}$ of database $D_i$\;
    }

    \Return $(Q', D')$\;

\caption{\TRIM for MAX}
\label{alg:max-trim}
\end{algorithm}

\subsection{Proof of \Cref{lem:lex_trim}}
Let $\wval = (\wval_1, \ldots, \wval_r)$.
The proof is the same as in the case of \rankmin/\rankmax (\Cref{sec:proof_minmax}),
except that the conditions we enforce in the $i^\textrm{th}$ of the $r$ copies of the database $D$
are
$P_i = \{ \w_{x_1}'(x_1) = \lambda_1, \ldots, \w_{x_{i-1}}'(x_{i-1}) = \lambda_{i-1}, 
                \w_{x_i}'(x_i) < \lambda_i \}$
for $\leq_\textrm{LEX}$
and
$P_i = \{ \w_{x_1}'(x_1) = \lambda_1, \ldots, \w_{x_{i-1}}'(x_{i-1}) = \lambda_{i-1}, 
                \w_{x_i}'(x_i) > \lambda_i \}$
for $\geq_\textrm{LEX}$.

\subsection{Proof of \Cref{th:partial_sum}}

First, we prove that the condition in our dichotomy
is equivalent to having the \ranksum variables
on one or two adjacent join tree nodes.

\begin{lemma}
\label{lem:adjacent_join_tree}
Consider the hypergraph $\calH(Q)$ of a \JQ $Q$
and a set of variables $\Wvars$. If $\calH(Q)$ is acyclic, 
any set of independent variables of $\Wvars$ is of size at most $2$, 
and any chordless path between two $\Wvars$ variables is of length at most $3$, then there exists a join tree for $Q$ where $U$ appears on one or two adjacent nodes.
\end{lemma}
\begin{proof}

If there is one query atom that contains all $\Wvars$ variables, then we are done. 
Otherwise, since any set of independent variables of $\Wvars$ is of size at most $2$, 
then there are $2$ atoms 
that together contain all $\Wvars$ variables. 
Indeed, consider any $3$ atoms. 
If each of them has a $\Wvars$ variable that does not appear in the other two, 
then these three variables are an independent set
of size $3$, which contradicts our condition. 
Thus, $2$ of these atoms contain all $\Wvars$ variables that appear in the $3$ atoms. 
By applying this repeatedly to the selected $2$ atoms and an untreated atom until all atoms are treated, 
we get $2$ atoms that contain all of $\Wvars$ variables.

Since $Q$ is acyclic, it has a join tree. Let $R'$ and $S'$ be two join-tree nodes that together 
contain all of $\Wvars$.
Consider the path $P'$ from $R'$ to $S'$ in the join tree.
Let $R$ be the last node on $P'$ that contains all $\Wvars$ variables that are in $R'$, 
and let $S$ be the first node on $P'$ that contains all $\Wvars$ variables that are in $S'$.
If $R$ and $S$ are neighbors, we are done.
Otherwise, we show we can find an alternative join tree where they are neighbors.
Consider the path $P$ from $R$ to $S$ in the join tree.
Let $V$ be all the variables that appear 
on the path between $R$ and $S$ (not including $R$ and $S$),
such that each variable in $V$ appears in either $R$ or $S$ (or both). 
We consider three cases.
The first case is $V \subseteq R$.
We directly connect $R$ and $S$ and remove the edge connecting $S$ to the node preceding it on the path from $R$.
The running intersection property is maintained as for each variable, the nodes containing this variable remain connected.
The second case is $V \subseteq S$. It is handled similarly by directly connecting $R$ to $S$ and
removing the edge from $R$ to its succeeding node on the path to $S$.
The third case is that a variable $u \in V$ appears in $R$ but not in $S$ and another variable $v \in V$ appears in $S$ but not in $R$.
Since $R$ is the last in $P$ to contain all $\Wvars$ variables of $R'$, there is a variable $x \in \Wvars$ that appears in $R$ but nowhere else in $P$. 
Similarly, there is a variable $y \in \Wvars$ that appears in $S$ and nowhere else in $P$.
If every two consecutive nodes on $P$ share a variable, then we have a chordless path $x-u-\ldots-v-y$ of length at least $4$, contradicting our condition. 
Otherwise, we remove the edge between the two nodes that do not share a variable, and add an edge between $R$ and $S$, which preserves the running intersection property.
\end{proof}

We now show the dichotomy of \Cref{th:partial_sum}.

For the positive side, we apply
\Cref{lem:adjacent_join_tree}.
When all $\Wvars$ variables are contained in a single join-tree node,
trimming can be done in linear time by filtering the corresponding relation.
When they are contained in two adjacent join-tree nodes,
$\bigO(n \log n)$ trimming follows from \Cref{lem:partial_sum_trim}.
Combining these two cases with \Cref{lem:quantile_exact,lem:pivot} 
completes the proof of the positive side.

For the negative side, there are $3$ cases.
1) If $Q$ is cyclic, an answer to \%\JQ would also answer the 
decision problem of whether $Q$ has any answer, which precludes time $\bigO(n \polylog n)$ assuming \hyperclique{}~\cite{bb:thesis}.
Assume $Q$ is acyclic.
2) If there exists a set of independent variables of $\Wvars$ of size $3$, 
selection by \ranksum is not possible in $\bigO(n^{2-\varepsilon})$ for all $\varepsilon>0$ assuming \ThreeSUM{}~\cite[Corollary 7.11]{carmeli23direct}. 
Since we can count the answers to an acyclic \JQ in linear time, the selection problem and \%\JQ are equivalent.
3) If there is a chordless path between two $\Wvars$ variables of length $4$ or more, we apply a known reduction~\cite[Lemma 7.13]{carmeli23direct} to show that solving \%\JQ in quasilinear time can be used to detect
a triangle in a graph in quasilinear time, which is not possible assuming \hyperclique{}. 
There are two ways in which the statement of that lemma differs from our needs: 
first, all variables there were allowed to participate in the ranking.
However, the reduction only assigns non-zero weights to the first and last variables in the path, so this difference is non-essential. 
Second, the path there contains exactly $3$ atoms (i.e., $4$ variables); 
if our path is longer, we simply make the the remaining relations equality, and the rest of the proof is the same.

\section{Details of Approximate Trimming for \ranksum}
\label{sec:approx_details}

\Cref{alg:sum_trim} shows the pseudocode of our lossy trimming for \ranksum.

\subsection{Proof of \Cref{lem:approx_trim}}
\label{sec:proof_approx_trim}

\begin{algorithm}[t]
\footnotesize
\textbf{Input}: acyclic \JQ $Q$ with $\ell$ atoms, database $D$, predicate $\sum_{x \in \Wvars} \w_x(x) < \wval$, approximation bound $\epsilon$\\
\textbf{Output}: acyclic \JQ $Q'$, database $D'$\\

Convert attribute weights to tuple weights\;
Construct a binary join tree $T$ of $Q$, set an arbitrary root\;
Materialize a relation for every $T$-node and group it by the variables it
has in common with its parent node\;
Initialize $\sumset(t) = (\sigma_s(t), \sigma_m(t)) = (\w(t), 1)$ for all tuples $t$ of all relations\; 
$\epsilon' = \frac{1}{4^\ell} \epsilon$\;

\For{relation $R$ in bottom-up order of $T$}
{
    \For{child $S$ of $R$}
    {Add variable $v_{RS}$ to $R$ and $S$ in $Q$, and corresponding columns in $D$\;
    }
    \For{tuple $t \in R$, child $S$ of $R$}
    {
        $b$ = join group of $S$ that agrees with the values of $t$\;
        \algocomment{Sketch messages the first time we visit this group}\;
        \If{$\sumset(b)$ not already computed }
        {
            $\sumset(b) = \mathbb{S}_{\epsilon'}(\cup_{t' \in b} \sigma(t'))$ such that each value falls into a single bucket \label{alg_line:sketch}\;

            \algocomment{A bucket $e$ in the sketch is described by a sum $e_s$, multiplicity $e_m$, and a set of source tuples from $S$}\;
            \For{bucket $e \in \sumset(b)$ with source tuples $S_e \subseteq S$}
            {
                \algocomment{Add the bucket values to the child column}\;
                 $t_e[v_{RS}] = (e_s, e_m)$ for all $t_e \in S_e$\;
            }
        }

        \For{bucket $e \in \sumset(b)$}
        {
            \algocomment{Add the bucket values to the parent column}\;
             Create a copy $t_e$ of $t$ in $R$ with $t_e[v_{RS}] = (e_s, e_m)$\;
             $\sigma(t_e) = (\sigma_s(t) + e_s, \sigma_m(t) \times e_m)$
             \label{alg_line:sum_update}\;
        }
        Remove $t$ from $R$\;
    }
}

Remove all tuples $t$ from the root relation with $\sigma_s(t) \geq \wval$ \label{alg_line:root_removal}\;

\Return $Q, D$\;

\caption{Approximate \TRIM for \ranksum}
\label{alg:sum_trim}
\end{algorithm}

\introparagraph{Preservation of \JQ answers}
Let $Q'$ and $D'$ be the returned \JQ and database.
We argue that, before removing
the root tuples that violate the inequality (\Cref{alg_line:root_removal}),
the \JQ answers are preserved
in the sense that there exists a bijection from $Q'(D')$ to
$Q(D)$ which simply removes the new variables.
Consider the step where we introduce variable $v_{RS}$ between parent $R$
and child $S$.
Let $t \in R$ be a tuple in the original database $D$
and $b$ the join group in $S$ that agrees with $t$.
Then, every tuple $t' \in b$ joins with exactly one copy of $t$ after
the introduction of $v_{RS}$.
This is because there is a copy of $t$ for each bucket (with the bucket identifier in $v_{RS}$)
and our bucket adjustment guarantees that
the weight of $t'$ is assigned to precisely one bucket.

\introparagraph{Error from sketch adjustment}
Recall that in our sketch $\mathbb{S}_\epsilon(L)$ of a multiset $L$
we made the adjustment
that if $i, j, k$ are three consecutive indexes in the bucketization,
$L[j-1] = L[j]$, and $j', j''$ are the smallest and largest indexes that contain $L[j]$
in the two consecutive buckets,
then we replace $i, j, k$ with $i, j', j''+1, k$.
We say that a multiset is an $\epsilon$-sketch of another multiset if it satisfies the guarantee of \Cref{lem:sketch}.
Also, let $\mathbb{S}$ be the original sketch with approximation error $\epsilon$
(see \Cref{lem:sketch})
and $\mathbb{S'}$ be the resulting sketch.
What we will show is that $\mathbb{S'}$ is an $\epsilon$-sketch of $L$.
In particular, we claim that 
$\countlessthan{\mathbb{S}}{\wval} \leq \countlessthan{\mathbb{S'}}{\wval} \leq \countlessthan{L}{\wval}$
for all values of $\wval$.
Our adjustment can only change elements in the index ranges $[i,j')$ and $[j, j'')$,
while all other elements stay the same since the largest element in their bucket continues to be the same.
The elements that can potentially change may only decrease in value because the upper index of their bucket is now smaller (but they may not decrease beyond $L[i]$ and $L[j]$ respectively).
Consequently, if $\wval < L[i]$ or $\wval \geq L[j'']$ then $\countlessthan{\mathbb{S}}{\wval} = \countlessthan{\mathbb{S'}}{\wval}$.
If $L[i] \leq \wval < L[j]$, then $\countlessthan{\mathbb{S}}{\wval} \leq \countlessthan{\mathbb{S'}}{\wval}$
because all elements in this bucket were mapped to $Z[j]$ in $\mathbb{S}$ but now they are
mapped to a number that can only be smaller, and thus closer to their original value.
If $L[j] \leq \wval < L[j']$, all elements in that bucket are equal to $L[j]$, 
thus $\countlessthan{\mathbb{S'}}{\wval} = \countlessthan{L}{\wval}$.

\introparagraph{Approximation guarantee}
Let us introduce the notation and tools we need.
Recall that each tuple $t$ computes $\sigma(t)$
that represents the approximate sum of partial query answers in its subtree.
Let $\textrm{cp}(t) = \{ t_1, \ldots, t_r \}$ be the copies of $t$ that we create in our algorithm,
$W_t$ be the partial query answers in the subtree of $t$
mapped to their weights,
$\textrm{jg}_S(t)$ be the join group of relation $S$ that joins with a tuple $t$ of the parent relation,
and $\otimes$ be the pairwise summation operator for multisets.
Abo-Khamis et al.~\cite{khamis21approx} have shown that
if $L_1'$ is an $\epsilon_1$-sketch of $L_1$
and $L_2'$ is an $\epsilon_2$-sketch of $L_2$,
then $L_1' \uplus L_2'$ is a $\max\{\epsilon_1,\epsilon_2\}$-sketch of $L_1 \uplus L_2$
and $L_1' \otimes L_2'$ is an $(\epsilon_1+\epsilon_2)$-sketch of $L_1 \otimes L_2$.
Additionally, an $\epsilon_1$-sketch of an $\epsilon_2$-sketch
is a $(2 \max\{\epsilon_1,\epsilon_2\})$-sketch
(using the definition and that $(1-\epsilon)^2 \geq 1-2 \epsilon$).
With these, we will show that the removal of root-node tuples (\Cref{alg_line:root_removal})
removes the \JQ answers that fall into buckets with values greater than or
equal to $\lambda$ in an $\epsilon$-sketch of the multiset 
$\{ \w(q) | q \in Q(D) \}$.
Note that in the algorithm, we apply sketching with
$\epsilon' \leq \epsilon$
(\Cref{alg_line:sketch}).

First, we prove inductively that
for a tuple $t \in R$ where $R$ is a relation at level $d$
(i.e., the maximum-length path from $R$ to a leaf node is $d$),
 $\uplus_{t_i \in \textrm{cp}(t)} \sigma(t_i)$ is a
 $(4^d \epsilon')$-sketch of $W_t$.
Each weight in $W_t$ is the sum of the weight of $t$
and the weights of the joining partial answers
(in the original database $D$)
from the child relations,
i.e.,
 $W_t = \{\w(t)\} \otimes (\bigotimes_{S} 
\bigcup_{t' \in \textrm{jg}_S(t)} W_{t'}))$.
If $t$ joins with a tuple $t'$ of a child relation $S$,
then it needs to join with all copies $\textrm{cp}(t')$
that were created when we handled $S$ and its children. 
The algorithm computes the values $\sigma(t_i)$ as follows:
$\uplus_{t_i \in \textrm{cp}(t)} \sigma(t_i) = \{\w(t)\} \otimes
(\bigotimes_{S} 
\mathbb{S}_{\epsilon'}(
\biguplus_{t' \in \textrm{jg}_S(t)} \biguplus_{t_j' \in \textrm{cp}(t')} \sigma(t_j')
)
)$.
We know inductively that $\biguplus_{t_j' \in \textrm{cp}(t')} \sigma(t_j')$
is a $(4^{d-1} \epsilon')$-sketch of $W_{t'}$.
The error bound of the sketch remains $4^{d-1} \epsilon'$
after the union,
then becomes $2 \cdot 4^{d-1} \epsilon'$ after applying the $\epsilon'$-sketch,
and finally $2 \cdot 2 \cdot 4^{d-1} \epsilon' = 4^{d} \epsilon'$ after taking the pairwise sums between
the two children.

Second, we claim that the height of the binary join tree we construct
is no more than $\ell$, where $\ell$
is the number of atoms of $Q$.
To see why, note that the new nodes we introduce in order to make the tree binary cannot be leaves and will always have $2$ children.
Suppose that there exists a root-to-leaf path of length greater than $\ell$.
For every new node on the path, there must be an original node that is a descendant of it, but not on this path.
This implies that the number of original nodes would be greater than $\ell$, which is a contradiction.
To conclude, we get $\epsilon$-sketches of $W_t$ for tuples $t$ at the root level if
we set $\epsilon' = \frac{1}{4^\ell} \epsilon$.
Their union is an $\epsilon$-sketch of $\{ \w(q) | q \in Q(D) \}$.

\introparagraph{Returned \JQ properties}
The fact that the \JQ $Q'$ that we return is acyclic is evident from the fact
that every variable $v_{RS}$ that we introduce appears in two adjacent 
nodes of the join tree of $Q$.
Therefore, $Q'$ also has a join tree.

\introparagraph{Running time}
The size of any relation is initially bounded by $n$.
Consider the step where we handle a relation and increase its size by creating copies of its tuples.
The sizes of the child relations (which have already been handled) have size bounded by $n' \geq n$.
The total size of the messages sent from the children
is $\bigO(\log_{1+\epsilon'} n')$
because the messages are sketched. 
The parent relation receives the messages of a child and 
creates copies of its tuples whose number is
equal to the message size.
Since we have at most 2 children, the size of the parent relation becomes $\bigO(n (\log_{1+\epsilon'} n')^2)$.
Applying this for every relation bottom-up, we can conclude that all relations after the algorithm terminates 
have size $\bigO(n (\log_{1+\epsilon'} n)^2)$
(because the double-logarithmic terms are dominated).
Changing base, this is
$\bigO(n \frac{\log^2 n}{\log^2 (1+\epsilon')})$
or $\bigO(\frac{1}{\epsilon^2} n \log^2 n)$
since $\epsilon' = \Theta(\epsilon)$
and also
$\log(1+\epsilon)$ is very close to $\epsilon$ for small $\epsilon$.
All other operations of the algorithm are linear in this size,
except for sketching, which is only done once for each join group.
A sketch of a multiset $L = (Z,\beta)$ can be computed in $\bigO(|Z| \log |Z|)$
by sorting.
Since $\bigO(\log(\frac{1}{\epsilon^2} n \log^2 n)) = \bigO(\log \frac{n}{\epsilon})$, we get the desired time bound
$\bigO(\frac{1}{\epsilon^2} n \log^2 n \log \frac{n}{\epsilon})$.

\end{document}